\documentclass[journal,draftcls,onecolumn,12pt,twoside]{IEEEtranTCOM}
\renewcommand{\baselinestretch}{1.5}

\usepackage{moreverb}
\usepackage{epsfig}
\usepackage{amsmath,amssymb,amsthm,mathrsfs,amsfonts,dsfont}
\usepackage{adjustbox,lipsum}
\usepackage{algorithm,algorithmic}
\usepackage{amsfonts}
\usepackage{epsfig}
\usepackage{amssymb}
\usepackage{amsmath}
\usepackage{amsthm}
\usepackage{multirow}
\usepackage{rotating}
\usepackage{graphicx}
\usepackage{tabularx}
\usepackage{array}
\usepackage{anyfontsize}
\usepackage{color,soul}
\usepackage{graphicx,dblfloatfix}
\usepackage{epstopdf}
\usepackage{blindtext}
\usepackage{amsmath}
\usepackage{amsthm,amssymb,amsmath,bm}
\usepackage{subfig}
\usepackage{amsfonts}
\usepackage{epsfig}
\usepackage{amssymb}
\usepackage{amsmath}
\usepackage{cite}
\usepackage{graphicx}
\usepackage{fancyhdr}
\usepackage[font={small}]{caption}
\usepackage{tabularx}
\usepackage{tcolorbox}
\usepackage{cite}
\usepackage{setspace}

\newtheorem{theorem}{Theorem}
\newtheorem{lemm}{Lemma}

\newtheorem{coro}{Corollary}

\begin{document}
\allowdisplaybreaks
\setlength{\abovedisplayskip}{6pt}
\setlength{\belowdisplayskip}{6pt}

\title{On the Age of Information in  Multi-Source  Queueing Models
\thanks{Mohammad Moltafet and Markus Leinonen are with the Centre
for Wireless Communications--Radio Technologies, University of Oulu,
90014 Oulu, Finland (e-mail: mohammad.moltafet@oulu.fi; markus.leinonen@oulu.fi), and
Marian Codreanu is with Department of Science and Technology, Link\"{o}ping University, Sweden (e-mail: marian.codreanu@liu.se)
}
\author{
Mohammad~Moltafet, Markus~Leinonen, and Marian~Codreanu
}}
\maketitle
\begin{abstract}
Freshness of status update packets is essential for enabling services where a destination needs the most recent measurements of  various sensors. In this paper, we study the information freshness of single-server multi-source queueing models under a first-come first-served (FCFS) serving policy. In the considered model, each source independently generates status update packets according to a Poisson process. The information freshness of the status updates of each source is evaluated by the average age of information (AoI). We derive an exact expression for the average AoI for the case with exponentially distributed service time, i.e., for a multi-source M/M/1 queueing model. Moreover, we derive three approximate expressions for the average AoI for a multi-source M/G/1 queueing model having a general service time distribution. Simulation results are provided to validate the derived exact average AoI expression, to assess the tightness of the proposed approximations, and to demonstrate the AoI behavior for different system parameters.

\end{abstract}

\section{Introduction}\label{Introduction}
Recently, various services in wireless sensor networks (WSNs) such as Internet of Things  and cyber-physical control applications have attracted both academic and industrial attention. In these networks, low power sensors may be assigned to  send status updates about a random process to intended  destinations \cite{corke2010environmental,5307471,818826,hu2002ensuring,8187436,8469047}. Such a status update system can monitor, e.g., temperature of a specific environment (room, greenhouse, etc.) \cite{corke2010environmental}, and a vehicular status (position, acceleration, etc.) \cite{5307471}. One key enabler  for these services is high freshness of the sensors' information at a destination. For instance, real-time control and decision making in the system requires that the destination has very recent measurements of the various sensors.

The traditional  metrics such as throughput and delay can not fully characterize the  information freshness \cite{6195689,8187436,8469047}.
Recently, the age of information (AoI) was proposed as a destination-centric  metric   to measure the information freshness \cite{6195689,6310931,5984917} in status update systems.
A status update packet contains the measured value of a monitored process and a time stamp representing the time when the sample was generated. Due to wireless channel access, channel errors, and fading etc., communicating a status update packet through the network experiences a random delay. If at a time instant $t$, the most recently received status update packet contains the time stamp $U(t)$, AoI is defined
as the random process $\Delta(t)=t-U(t)$.
Thus, the AoI measures for each sensor   the time elapsed since the last received status update
packet was generated at the sensor.  The most common metrics of the AoI are   average AoI,  peak AoI, and effective AoI   \cite{8187436,8006593,6875100}. In this work, we focus on the average AoI.

\subsection{Related Works}
The first queueing theoretic work on AoI is \cite{6195689} where the authors derived the average AoI for a single-source M/M/1 first-come first-served (FCFS) queueing model. The average  AoI for an  M/M/1 last-come first-served  (LCFS) queueing model with preemption was analyzed in \cite{6310931}.
In \cite{6875100}, the authors proposed peak AoI as an alternative  metric to evaluate the information freshness.
The average AoI  and average peak AoI for different packet management policies in an M/M/1 queueing model were derived in \cite{7415972}. The authors of \cite{8006504}  derived a closed-form expression for the average
AoI of a single-source  M/G/1/1 preemptive queueing model (where the last entry in the Kendall notation shows the total capacity of the queueing system; 1 indicates that there is one packet under service whereas the queue holds zero packets).
A closed-form expression for the average AoI in a single-source M/G/1 queueing model was derived in \cite{8006592}.
The work \cite{7541764} considered a single-source LCFS queueing model where the packets arrive according to a Poisson process and the service time follows a gamma distribution. They derived the average AoI and average peak AoI for two packet management policies, LCFS  with and  without preemption.

Besides single-source setups, the work \cite{6284003} was the first to investigate the average AoI in a multi-source setup. The authors of \cite{6284003} derived the average AoI for a multi-source M/M/1 FCFS queueing model.
The authors of \cite{7282742} considered a multi-source M/G/1 queueing system  and optimized the arrival rates of each source to minimize the peak AoI.
The closed-form expressions for the average AoI and average peak AoI in a multi-source M/G/1/1 preemptive queueing model were derived in \cite{8406928}.
In \cite{8469047}, the authors  introduced a powerful technique based on stochastic hybrid systems to evaluate the AoI in finite-state continuous-time queueing systems.

The  AoI has also been applied as a novel metric in various networking problems. The AoI in a carrier sense multiple access  (CSMA) based vehicular network was studied via simulations in \cite{5984917}. The authors of \cite{Kosta2018AgeOI} studied AoI and throughput in a shared access network having one primary and several secondary transmitter-receiver pairs.  The authors of \cite{8006544} investigated the AoI for  ALOHA and time-scheduled based access techniques in  WSNs. They concluded that ALOHA access,  while simple, leads to AoI that is inferior to a scheduled access case. The authors of \cite{8445979} considered a  WSN, derived the average AoI and peak AoI for the system, and minimized  the average AoI and  peak AoI    by optimizing the probability of transmission of each node.  The authors of \cite{9007478} analyzed the AoI in a CSMA based system using the stochastic hybrid systems technique.
 They optimized the system's average AoI by adjusting the back-off time of each link. The authors of \cite{8877239}  analyzed the worst case average AoI for each sensor in a CSMA based WSN.

\subsection{Contributions}
In this paper, we analyze the average AoI of the different sources in single-server multi-source queueing models under an FCFS service policy with Poisson packet arrivals.
First,  derive  an \emph{exact} expression for the average AoI for a multi-source M/M/1 queueing model. The setup was earlier addressed in \cite{6284003,8469047}, where the authors derived an approximate expression for the  average AoI by neglecting the statistical dependency between certain random variables (see Section \ref{Exact Expression for the Average AoI in a Multi-Source M/M/1 Queueing Model}). Second,  we point out the difficulties in an  M/G/1 case and  derive three \emph{approximate} expressions for the average AoI in  a multi-source M/G/1 queueing model. We present simulation results to 1) validate the derived exact average AoI in a multi-source M/M/1 queueing model, 2) show that the proposed approximations are relatively tight in both the M/M/1 case and the M/G/1 case where the service time follows different distributions, and 3) exemplify the AoI behavior under different system parameters.

\subsection{Organization}
This paper is organized as follows. The system model,   AoI definition, and a summary of the main results  are  presented in Section II.
The main steps required to derive the average AoI for a multi-source M/G/1 queueing model are presented in Section III.
The exact expression for the average AoI in a multi-source M/M/1 queueing model is derived in Section \ref{Exact Expression for the Average AoI in a Multi-Source M/M/1 Queueing Model}.
The three approximate expressions for the average AoI in a multi-source M/G/1 queueing model are derived in Section \ref{An Approximation for the Average AoI in a Multi-Source M/G/1 Queueing Model}.
Numerical validation and results are presented in Section \ref{Validation and results}.
Finally, concluding remarks are expressed in Section \ref{conclusion}.

\section{System Model and Summary of Results }
We  consider a system consisting of a set of independent sources denoted by $\mathcal{C}=\{1,\dots,C\}$ and one server, as depicted in Fig. \ref{Ao00101I}.
Each source observes a random process, representing, e.g., temperature, vehicular speed or location at random time instants. A remote destination is interested in timely information about the status of these random processes. Status updates are transmitted as packets, containing the measured value of the monitored process and a time stamp representing the time when the sample was generated. We assume that  the packets  of source  $c$  are generated according to the   Poisson process with rate $\lambda_c$, $c\in \mathcal{C}$.

For each source, the AoI  at the destination is defined as the time elapsed since the last  successfully received packet was generated. Formal definition of the AoI is given next.

\begin{figure*}
\centering
\includegraphics[width=0.53\linewidth,trim = 5mm 0mm 5mm 5mm,clip]{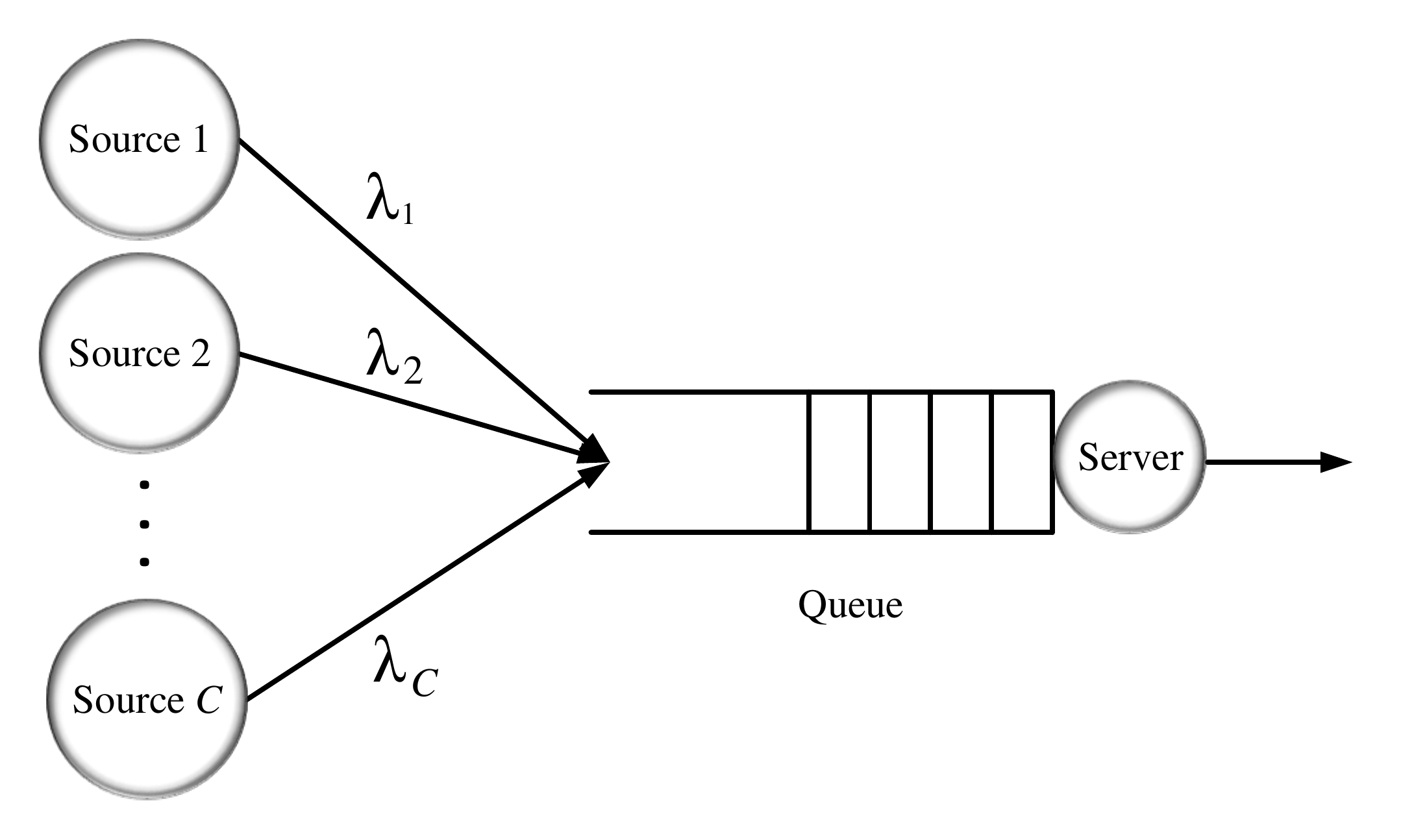}
\caption{The considered status update system modeled as a multi-source M/G/1 queueing model.
}\vspace{-12mm}
\label{Ao00101I}
\end{figure*}

\textbf{Definition 1} (AoI).
Let $t_{c,i}$ denote the time instant at which the $i$th status   update packet of source $c$ was generated, and $t'_{c,i}$ denote the time instant at which this packet  arrives at the destination. At a time instant $\tau$,  the  index of the most recently received packet of source $c$ is given by
\begin{equation}\label{mnb00}
N_c(\tau)=\max\{i'|t'_{c,i'}\le \tau\},
\end{equation}
and the time stamp of the most recently received packet of source $c$ is
$
U_c(\tau)=t_{c,N_c(\tau)}.
$
The AoI of source $c$ at the destination is defined as the random process
$
\Delta_{c}(t)=t-U_c(t).
$

An example of evolution of the AoI is shown in Fig. \ref{AoI}. As it can be seen, $ \Delta_{c}(t) $ at the destination increases linearly with time, until the
reception of a new status update, when the AoI is reset to the age of the newly received status update, i.e., the difference of the current time instant and the time stamp of the newly received update.

The   most commonly used  metric  for evaluating the AoI of a source at the destination is the average AoI  \cite{8187436,8006593,6875100}.  Next, we introduce this  metric for the considered system model.

\begin{figure}
	\centering
	\includegraphics[scale=.8]{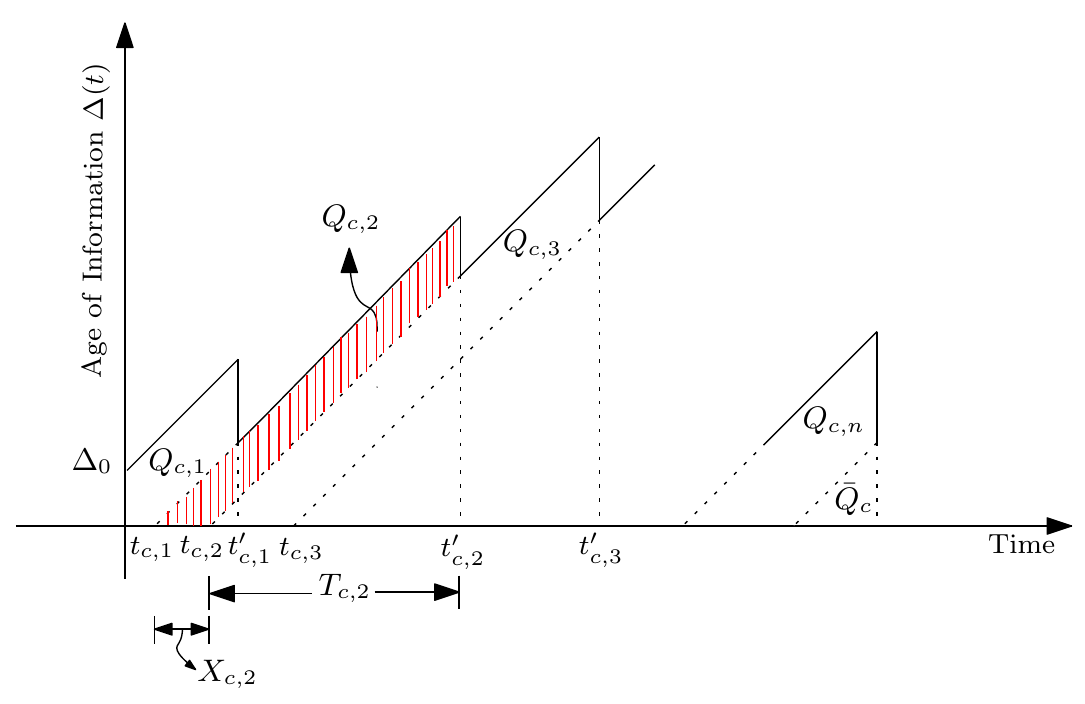}
	\vspace{-5mm}
	\caption{Age of information of source $c$ as a function of time.}
	\vspace{-13mm}
	\label{AoI}
\end{figure}
\subsection{Average AoI}
Let $(0,\tau)$ denote   an observation interval. Accordingly,  the time average AoI of  the source $c$ at the destination, denoted as  $\Delta_{\tau,c}$, is  defined as
\begin{equation}\label{oointr}
\Delta_{\tau,c}=\dfrac{1}{\tau}\int_{0}^{\tau}\Delta_{c}(t)\mathrm{d}t.
\end{equation}
The integral in \eqref{oointr} is equal to  the area under $\Delta_{c}(t)$ which can be  expressed  as a sum of disjoint areas determined by  a polygon
$ Q_{c,1} $, $ N_c(\tau)-1 $ trapezoids  $Q_{c,i}, i=2,\dots,N_c(\tau)$, and  a triangle  $\bar{Q}_c$, as illustrated in Fig. \ref{AoI}.
Following the definition of  $N_c(\tau)$ in \eqref{mnb00}, $\Delta_{\tau,c}$ can be calculated as
\begin{equation}\label{oointr01}
\Delta_{\tau,c}=\dfrac{1}{\tau}\bigg(Q_{c,1}+\textstyle\sum_{i=2}^{N_c(\tau)}Q_{c,i}+\bar{Q}_c\bigg)=\dfrac{Q_{c,1}+\bar{Q}_c}{\tau}+\dfrac{N_c(\tau)-1}{\tau}\dfrac{1}{N_c(\tau)-1}\textstyle\sum_{i=2}^{N_c(\tau)}Q_{c,i}.
\end{equation}

The  average AoI of source $c$, denoted by $\Delta_{c}$, is defined as
$
\Delta_{c}=\lim_{\tau\to\infty}\Delta_{\tau,c}.
$
 The term $\dfrac{Q_{c,1}+\bar{Q}_c}{\tau}$ in \eqref{oointr01} goes to zero as $\tau\to\infty$, and the term $\dfrac{N_c(\tau)-1}{\tau}$ in \eqref{oointr01}  converges to the  rate of generating the status update packets of source $c$ as ${\tau\to\infty}$, i.e., $\lambda_c=\lim_{\tau\to\infty}\dfrac{N_c(\tau)-1}{\tau}$.  Moreover, as  ${\tau\to\infty}$, the number of transmitted packets grows
to infinity, i.e., ${N_c(\tau)\to\infty}$.
{Thus, assuming that the random process $ \{Q_{c,i}\}_{i>1} $ is (mean) ergodic\footnote{{Note that for the ergodicity assumption, it is necessary to have a stationary and  stable system (for the stability condition it is sufficient to have  $\textstyle\sum_{c\in\mathcal{C}}\lambda_{c}< \mu,$ where $\mu$ is the mean service rate in the system).}} \cite{6195689,8187436,8469047},
%
%
 the sample average term $\dfrac{1}{N_c(\tau)-1}\textstyle\sum_{i=2}^{N_c(\tau)}Q_{c,i}$ in \eqref{oointr01}  converges to the stochastic average $\mathbb{E}[Q_{c,i}]$.} Consequently, $\Delta_{c}$ is given by
%
%
$$
\Delta_{c}=\lambda_c\mathbb{E}[Q_{c,i}].
$$

As shown in Fig. \ref{AoI},  $Q_{c,i}$ can be calculated by subtracting the area of the isosceles triangle with sides $(t'_{c,i}-t_{c,i})$ from the area of the isosceles triangle with sides $(t'_{c,i}-t_{c,i-1})$. Let the random variable
\begin{equation}\label{mn120}
X_{c,i}=t_{c,i}-t_{c,i-1}
\end{equation}
represent the $i$th interarrival time of source $c$, i.e., the time elapsed between the generation of $i-1$th packet and $i$th packet from source $c$. From here onwards, we refer to the $i$th packet from source $c$ simply as packet $c,i$. Moreover, let the random variable
\begin{equation}\label{mn1201}
T_{c,i}=t'_{c,i}-t_{c,i}
\end{equation}
represent the system time of packet $c,i$, i.e., the time interval the packet spends in the system which consists of the sum of the  waiting time and the service time. 
	 By using \eqref{mn120} and \eqref{mn1201},  $Q_{c,i}$ can be calculated by subtracting the area of the isosceles triangle with sides $X_{c,i}$ from the area of the isosceles triangle with sides $X_{c,i}+T_{c,i}$ (see Fig. \ref{AoI}), and thus, the average AoI of source $c$ is given as{\cite{6284003}}
\begin{equation}\label{oointr001003}
\Delta_{c}=\lambda_c\mathbb{E}[Q_{c,i}]=\lambda_c\bigg(\dfrac{1}{2}\mathbb{E}[(X_{c,i}+T_{c,i})^2]-\dfrac{1}{2}\mathbb{E}[X_{c,i}^2]\bigg)
=\lambda_c\bigg(\dfrac{\mathbb{E}[X_{c,i}^2]}{2}+\mathbb{E}[X_{c,i} T_{c,i}]\bigg).
\end{equation}

Let $W_{c,i}$ be the random variable representing the waiting time of packet $c,i$, and $S_{c,i}$ the random variable representing the service time of packet $c,i$. Consequently, the system time $T_{c,i}$ is given as the sum $T_{c,i}=W_{c,i}+S_{c,i}$, and the average AoI in \eqref{oointr001003} can be written as
\begin{equation}\label{oointr00103}
\Delta_{c}
=\lambda_c\bigg(\dfrac{\mathbb{E}[X_{c,i}^2]}{2}+\mathbb{E}[X_{c,i} (W_{c,i}+S_{c,i})]\bigg).
\end{equation}

{\subsection{Summary of the Main Results}\label{Summary of the Main Results}
	Here, we briefly summarize the main results of the paper. To evaluate the AoI of one source in a queueing model with multiple sources of Poisson arrivals, we can consider two sources without loss of generality. Thus, we proceed to evaluate the AoI of source 1 by aggregating the other ${C-1}$ sources into source 2 having the Poisson arrival rate ${\lambda_2=\textstyle\sum_{c'\in\mathcal{C}\setminus\{1\}}\lambda_{c'}}$. The mean service time for each packet in the system is equal, given as $\mathbb{E}[S_{1,i}]=\mathbb{E}[S_{2,i}]={1}/{\mu}$, $\forall{i}$. Let  $\rho_1={\lambda_1}/{\mu}$ and $\rho_2={\lambda_2}/{\mu}$ be the  load of  source 1 and 2, respectively. Since packets of each source are generated  according to the Poisson process and the sources are independent, the  packet generation in the system follows the Poisson process with rate
	$\lambda=\lambda_1+\lambda_2$, and the overall load in the system is
	$\rho=\rho_1+\rho_2={\lambda}/{\mu}$. Since we do not assume any specific probability density function (PDF) for the service time, the considered model is referred to a multi-source M/G/1 queueing model.}

{The main contributions of this paper are  twofold: we derive 1) an \textit{ exact} expression for the average AoI for a multi-source M/M/1 queueing model and 2) propose three \textit{approximate} expressions for the average AoI in  a multi-source M/G/1 queueing model. The derived results are summarized as follows.}

{\begin{theorem}\label{exactmm1theorem}
		The exact expression for the average AoI of source 1 for a multi-source M/M/1 queueing model is given in \eqref{mm10a0}  and has the following form:
		\begin{align}\nonumber
			\Delta_1&\!= \!\lambda^2_1(1-\rho)\Psi(\mu,\rho_1,\lambda_2)\!+\!\dfrac{1 }{\mu}\bigg(\dfrac{1 }{\rho_1}\!+\!
			\dfrac{\rho}{1-\rho}\!+\!\dfrac{(2\rho_2-1)(\rho - 1)}{(1-\rho_2)^2}+\dfrac{2\rho_1\rho_2(\rho - 1)}{(1-\rho_2)^3}\bigg),
		\end{align}
		where $\Psi(\mu,\rho_1,\lambda_2)$ is a function that is characterized by transient behavior of an M/M/1 queue which is presented in \eqref{01mnk}.
	\end{theorem}
	\textit{Proof:}
	The proof of Theorem \ref{exactmm1theorem} appears in parts in  Sections \ref{AoI_Multi} and \ref{Exact Expression for the Average AoI in a Multi-Source M/M/1 Queueing Model} of this paper.}

{The three approximate expressions for the average AoI of source 1 for a multi-source M/G/1 queueing model, denoted by $\Delta^{\text{app}_1}_{1}, \Delta^{\text{app}_2}_{1}$, and $\Delta^{\text{app}_3}_{1} $, 
	are given in \eqref{main542}, \eqref{main543}, and \eqref{main5433}, and are of the following form:
	\begin{align}\nonumber
		\Delta^{\text{app}_1}_{1}&\approx\mathbb{E}[W]+\dfrac{2}{\mu}+\dfrac{2\rho_{2}-1}{\lambda_1}+ \dfrac{2(1-\rho_{2})}{\lambda_1}L_{T}(\lambda_1)+(\rho_{2}-1)L'_{T}(\lambda_1).
	\end{align}
	\begin{align}\nonumber
		\Delta^{\text{app}_2}_{1}&\approx\mathbb{E}[W]+\dfrac{2}{\mu}+\dfrac{2\rho_{2}-1}{\lambda_1}+ \left(\dfrac{1}{\mu}+\dfrac{2(1-\rho_{2})}{\lambda_1}\right)L_{T}(\lambda_1)+\left(\rho_{2}-1-\dfrac{\lambda_1}{\mu}\right)L'_{T}(\lambda_1).
	\end{align}
	\begin{align}\nonumber
		\Delta^{\text{app}_3}_{1}&\approx\mathbb{E}[W]+\dfrac{2}{\mu}\!+\dfrac{2\rho_{2}-1}{\lambda_1}+ \left(\dfrac{\lambda_2\mathbb{E}[S^2]}{2(1-\rho_2)}+\dfrac{2(1-\rho_{2})}{\lambda_1}\right)L_{T}(\lambda_1)+\\&\nonumber\left(2\rho_{2}-1-\dfrac{\lambda_1\lambda_2\mathbb{E}[S^2]}{2(1-\rho_2)}\right)L'_{T}(\lambda_1)-\lambda_1\rho_2L''_{T}(\lambda_1),
	\end{align}
	where $ \mathbb{E}[W] $ is the average waiting time of each packet in the system (which is given in \eqref{queue000}),  $L_{T}(\lambda_1)$ is the Laplace transform of the PDF of the system time (which is given in \eqref{000bn0v001b0012}), and $L'_{T}(\lambda_1)$ and $L''_{T}(\lambda_1)$ are the first and second derivative of $L_{T}(\cdot)$ at $\lambda_1$.   
	The  calculations to derive the approximate expressions are presented in Sections  \ref{AoI_Multi} and \ref{An Approximation for the Average AoI in a Multi-Source M/G/1 Queueing Model} of this paper.}

\section{AoI in a Multi-Source M/G/1 Queueing Model}\label{AoI_Multi}

{In this section, we present the main steps required to derive the average AoI in (7) for the considered multi-source M/G/1 queueing model and point out the main difficulties regarding the average AoI calculation. Then, in Section  IV, we derive the  exact expression for the M/M/1 case and in Section V, we derive the approximate expressions for the M/G/1 case with a general service time distribution.}   


 The first term  in \eqref{oointr00103} is easy to compute. Namely, since  the  interarrival time of source $1$ follows the   exponential distribution with parameter $\lambda_1$, we have  $\mathbb{E}[X_{1,i}^2]=2/\lambda_1^2$. The second term in \eqref{oointr00103} can be written as
\begin{align}\label{mbn0}
\mathbb{E}[X_{1,i} (W_{1,i}+S_{1,i})]\stackrel{(a)}{=}\mathbb{E}[X_{1,i} W_{1,i}]+\mathbb{E}[X_{1,i}]\mathbb{E}[S_{1,i}]=\mathbb{E}[X_{1,i} W_{1,i}]+\dfrac{1}{\lambda_1\mu},
\end{align}
where equality (a) follows because  the interarrival time and service time   of the packet $1,i$  are independent. Since the random variables $X_{1,i}$ and $W_{1,i}$ are dependent, the  most challenging part in calculating \eqref{oointr00103} is $\mathbb{E}[X_{1,i} W_{1,i}]$ which is derived in the following.
%

In order to calculate $\mathbb{E}[X_{1,i} W_{1,i}]$, we follow the approach of  \cite{6284003} and characterize the waiting time $W_{1,i}$  by means of  two events $E^{\mathrm{B}}_{1,i}$ and $E^{\mathrm{L}}_{1,i}$ as
\begin{align}\label{0011nm}
&E^{\mathrm{B}}_{1,i}=\big\{T_{1,i-1}\ge X_{1,i}\big\},\quad E^{\mathrm{L}}_{1,i}=\big\{T_{1,i-1}<X_{1,i}\big\}.
\end{align}
Here, \textit{brief event} $E^{\mathrm{B}}_{1,i}$ is the event where the interarrival time of packet $1,i$ is brief, i.e., the interarrival time of packet $1,i$ is shorter than the system time of  packet $1,i-1$. On the contrary, \textit{long event} $ E^{\mathrm{L}}_{1,i} $ refers to the complementary event where  the interarrival time of packet $1,i$ is long, i.e., the interarrival time of packet $1,i$ is longer than the system time of  packet $1,i-1$.

Next, we characterize the waiting time for packet $1,i$.
Under the event $E^{\mathrm{B}}_{1,i}$, the waiting time of packet $1,i$ ($W_{1,i}$) contains two terms: 1) the residual system time to complete serving packet $1,i-1$, and 2) the sum of service times of the source 2 packets that arrived during $ X_{1,i} $ and   must be served before  packet $1,i$  according to the FCFS policy (see Fig. \ref{AoI021}(a)). Under the event $E^{\mathrm{L}}_{1,i}$, the waiting time of packet $1,i$ contains two terms: 1) the possible residual service time of a source 2 packet that is under service at the arrival instant of packet ${1,i}$, and 2) the sum of service times of source 2 packets in the queue that  must be served before packet $1,i$ according to the FCFS policy (see Fig. \ref{AoI021}(b)). 
{For the  event $E^{\mathrm{B}}_{1,i}$, let
	\begin{align}\label{mn0120}
	{R_{1,i}^{\mathrm{B}}=T_{1,i-1}-X_{1,i}}
	\end{align}
	represent the residual system time to complete serving packet $1,i-1$ and let
	\begin{align}\label{mn0121}
	S_{1,i}^{\mathrm{B}}=\textstyle\sum_{i'\in \mathcal{M}_{2,i}^{\mathrm{B}}}S_{2,i'}
	\end{align}
	represent the sum of service times of source 2 packets that arrived during $ X_{1,i} $ and  must be served before  packet $1,i$ where $\mathcal{M}_{2,i}^{\mathrm{B}}$ is  the set of indices of queued packets of source $2$ that must be served before packet $1,i$ under the event $E_{1,i}^{\mathrm{B}}$, where $|\mathcal{M}_{2,i}^{\mathrm{B}}|=M_{2,i}^{\mathrm{B}}$.
	Similarly for the event $E^{\mathrm{L}}_{1,i}$, let
	\begin{align}\label{mn0122}
	S_{1,i}^{\mathrm{L}}=\textstyle\sum_{i'\in \mathcal{M}_{2,i}^{\mathrm{L}}}S_{2,i'}
	\end{align}
	represent the sum of service times of source 2 packets that
	must be served before  packet $1,i$ where $\mathcal{M}_{2,i}^{\mathrm{L}}$ is  the set of indices of packets of source $2$ that are in the queue (but not under service) at the arrival instant of packet $1,i$ conditioned on the event $E_{1,i}^{\mathrm{L}}$ and, thus,  must be served before packet $1,i$, where  $|\mathcal{M}_{2,i}^{\mathrm{L}}|=M_{2,i}^{\mathrm{L}}$. Thus, by means of the two events in \eqref{0011nm} and  definitions \eqref{mn0120}, \eqref{mn0121}, and \eqref{mn0122}, the waiting time for packet $1,i$ can be expressed as
	\begin{align}\label{bnvb01}
	W_{1,i}&=\begin{cases}
	S_{1,i}^{\mathrm{B}}+R_{1,i}^{\mathrm{B}},&E^{\mathrm{B}}_{1,i}\\
	S_{1,i}^{\mathrm{L}}+R^{\mathrm{L}}_{2,i},&E^{\mathrm{L}}_{1,i},
	\end{cases}
	\end{align}
	where $R^{\mathrm{L}}_{2,i}$ is a random variable that represents the possible residual service time of the   packet of source 2 that is under service at the arrival instant of packet ${1,i}$ conditioned on the event  $E_{1,i}^{\mathrm{L}}$.}


%

\begin{figure}[t!]
	\centering
	\subfloat[]{\includegraphics[width=0.53\linewidth,trim = 0mm 0mm 0mm 0mm,clip]{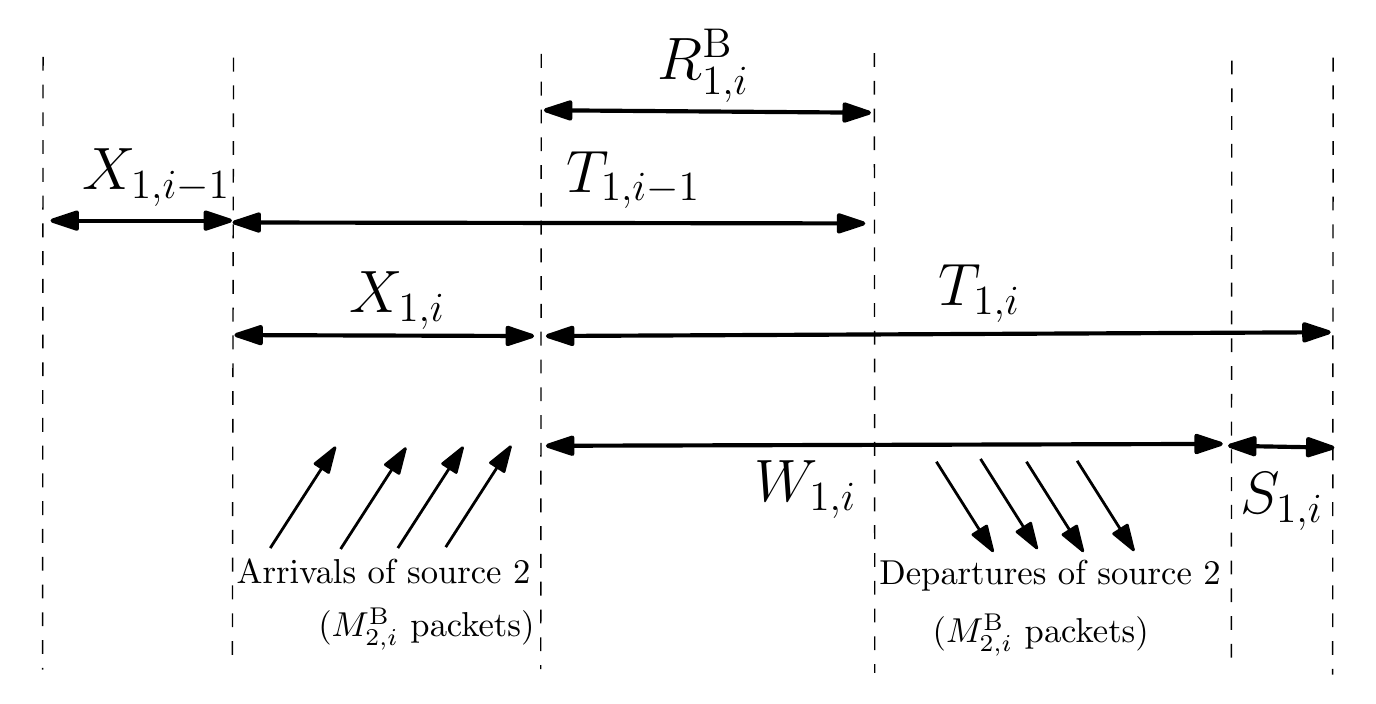}}\vspace{-3mm}\\
	\subfloat[]{\includegraphics[width=0.53\linewidth,trim = 0mm 0mm 0mm 0mm,clip]{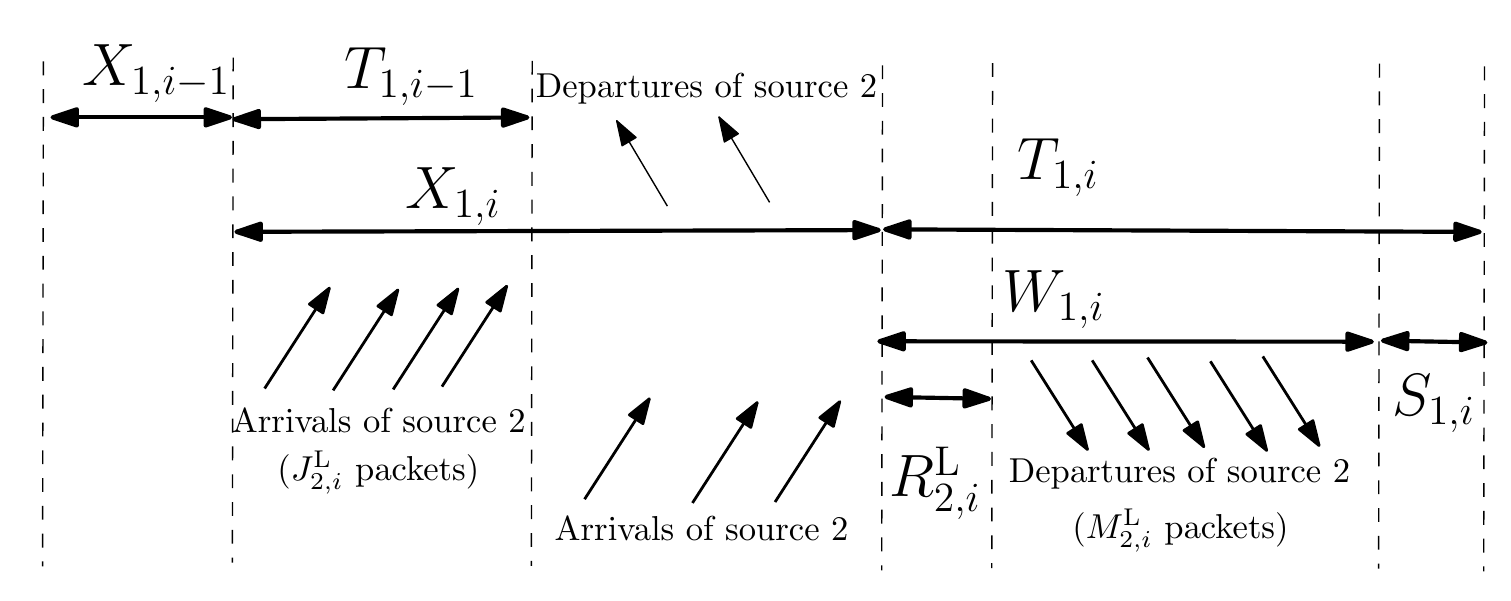}}\vspace{-3mm}
	\caption{ Illustration of the key quantities in characterizing the waiting time in \eqref{bnvb01} under (a) brief event $E^{\mathrm{B}}_{1,i}$  and  (b) long event $E^{\mathrm{L}}_{1,i}$. }\vspace{-10mm}  			
\label{AoI021}
\end{figure}


Based on \eqref{bnvb01},  $\mathbb{E}[X_{1,i} W_{1,i}]$ in \eqref{mbn0} can be expressed as
\begin{align}\label{mnbv0}
\mathbb{E}[X_{1,i}W_{1,i}]=& \bigg(\mathbb{E}[R_{1,i}^{\mathrm{B}}X_{1,i}|E^{\mathrm{B}}_{1,i}]+\mathbb{E}[S_{1,i}^{\mathrm{B}}X_{1,i}|E^{\mathrm{B}}_{1,i}]\bigg)P(E^{\mathrm{B}}_{1,i})+\\&\nonumber\mathbb{E}[(S_{1,i}^{\mathrm{L}}+R^{\mathrm{L}}_{2,i})X_{1,i}|E^{\mathrm{L}}_{1,i}]P(E^{\mathrm{L}}_{1,i}),
\end{align}
where $ P(E^{\mathrm{B}}_{1,i}) $ and $ P(E^{\mathrm{L}}_{1,i}) $ denote the probabilities of the events $ E^{\mathrm{B}}_{1,i}$ and $E^{\mathrm{L}}_{1,i}$, respectively.

Next, we derive the expressions for $P(E^{\mathrm{B}}_{1,i})$ and $P(E^{\mathrm{L}}_{1,i})$ in \eqref{mnbv0}. Then, by referring to $\mathbb{E}[R_{1,i}^{\mathrm{B}}X_{1,i}|E^{\mathrm{B}}_{1,i}]$, $\mathbb{E}[S_{1,i}^{\mathrm{B}}X_{1,i}|E^{\mathrm{B}}_{1,i}]$, and  $\mathbb{E}[(S_{1,i}^{\mathrm{L}}+R^{\mathrm{L}}_{2,i})X_{1,i}|E^{\mathrm{L}}_{1,i}]$ in \eqref{mnbv0} as the first, the second,  and the third conditional expectation terms of \eqref{mnbv0}, we present elaborate derivations of the first and second terms in Sections \ref{Label_for_subsubsection_1} and  \ref{Label_for_subsubsection_2}, respectively, and in Section \ref{Label_for_subsubsection_3} we point out the difficulties involved in computing the third  term for a generic service time distribution.

The following lemma gives the expressions for $P(E^{\mathrm{B}}_{1,i})$ and $P(E^{\mathrm{L}}_{1,i})$ in \eqref{mnbv0}.

\begin{lemm}\label{lemmii1}
The probabilities of the events $E_{1,i}^{\mathrm{B}}$ and $E_{1,i}^{\mathrm{L}}$ in \eqref{0011nm} are calculated as follows:
\begin{align}\label{proffe01}
P(E^{\mathrm{B}}_{1,i})=	\dfrac{L_{S}(\lambda_1)(\lambda+(\rho-1)\lambda_1)-\lambda_{2}}{\lambda L_{S}(\lambda_1)-\lambda_{2}},
\end{align}
\begin{align}\label{proffe02}
P(E^{\mathrm{L}}_{1,i})=\dfrac{(1-\rho)\lambda_1L_{S}(\lambda_1)}{\lambda L_{S}(\lambda_1)-\lambda_{2}},
\end{align}
{where $L_{S}(\lambda_1)$ is the Laplace transform of the PDF of the   service time $S$  at $\lambda_1$}; note that the service times of all packets are stochastically identical as $S_{1,i}=^{\mathrm{st}}S_{2,i}=^{\mathrm{st}}S$, $\forall{i}$.
\end{lemm}

\begin{proof}
See Appendix \ref{Lemapp}.
\end{proof}

\subsubsection{The First Conditional Expectation in \eqref{mnbv0}}\label{Label_for_subsubsection_1}
Let us now focus on the first conditional expectation term $\mathbb{E}[R_{1,i}^{\mathrm{B}}X_{1,i}|E^{\mathrm{B}}_{1,i}]$ in \eqref{mnbv0}. According to \eqref{mn0120},  this term is expressed as follows:
\begin{align}\label{vbfgh}
\mathbb{E}[R_{1,i}^{\mathrm{B}}X_{1,i}|E^{\mathrm{B}}_{1,i}]&= \mathbb{E}[T_{1,i-1}X_{1,i}|E^{\mathrm{B}}_{1,i}]-\mathbb{E}[X^2_{1,i}|E^{\mathrm{B}}_{1,i}]\\&\nonumber
=\int_{0}^{\infty}\int_{0}^{\infty} xt f_{X_{1,i},T_{1,i-1}|E^{\mathrm{B}}_{1,i}}(x,t)\mathrm{d}x \mathrm{d}t- \int_{0}^{\infty} x^2 f_{X_{1,i}|E^{\mathrm{B}}_{1,i}}(x) \mathrm{d}x,
\end{align}
where $f_{X_{1,i}|E^{\mathrm{B}}_{1,i}}(x)$ is the conditional PDF of the interarrival time $X_{1,i}$ given the event $E^{\mathrm{B}}_{1,i}$  and  $f_{X_{1,i},T_{1,i-1}|E^{\mathrm{B}}_{1,i}}(x,t)$ is the conditional joint PDF of the interarrival time $X_{1,i}$ and system time $T_{1,i-1}$ given the event $E^{\mathrm{B}}_{1,i}$. They are given by the following  lemma and corollary.

{\begin{lemm}\label{gh0fd0}
The conditional PDF $f_{X_{1,i},T_{1,i-1}|E^{\mathrm{B}}_{1,i}}(x,t)$ is given by
\begin{align}\label{nvhg012}
f_{X_{1,i},T_{1,i-1}|E^{\mathrm{B}}_{1,i}}(x,t) =
\begin{cases}0&x>t\\\dfrac{\lambda_1e^{-\lambda_1 x}f_{T_{1,i-1}}(t) }{P(E^{\mathrm{B}}_{1,i})}&x\le t.\end{cases}
\end{align}	
\end{lemm}
\begin{proof}	
See Appendix \ref{Lemapp}.
\end{proof}}
{The conditional PDF $f_{X_{1,i}|E^{\mathrm{B}}_{1,i}}(x)$ is determined by the following corollary, which is an immediate consequence of Lemma 2.
	\begin{coro} \label{ghfd0}
   The conditional PDF $f_{X_{1,i}|E^{\mathrm{B}}_{1,i}}(x)$ is given by
   \begin{align}\label{hgjr0}
   f_{X_{1,i}|E^{\mathrm{B}}_{1,i}}(x)=\dfrac{\lambda_1 e^{-\lambda_1 x}(1-F_{T_{1,i-1}}(x))}{P(E^{\mathrm{B}}_{1,i})},
   \end{align}
   where $F_{T_{1,i-1}}(x)$ is the cumulative distribution function  of  $T_{1,i-1}$.
\end{coro}
}

Now, having introduced the conditional PDFs   in  Lemma  \ref{gh0fd0} and Corollary \ref{ghfd0}, we can compute the conditional expectation $\mathbb{E}[R_{1,i}^{\mathrm{B}}X_{1,i}|E^{\mathrm{B}}_{1,i}]$ in  \eqref{vbfgh}.
Using Lemma \ref{gh0fd0}, the first  term  in \eqref{vbfgh} is calculated as
\begin{align}\label{102mn}
\mathbb{E}[T_{1,i-1}X_{1,i}|E^{\mathrm{B}}_{1,i}] &=\int_{0}^{\infty}\int_{0}^{\infty} xt f_{X_{1,i},T_{1,i-1}|E^{\mathrm{B}}_{1,i}}(x,t)\mathrm{d}x \mathrm{d}t\\&\nonumber
= \dfrac{1}{P(E^{\mathrm{B}}_{1,i})}\int_{0}^{\infty}\int_{0}^{t} tx \lambda_1 e^{-\lambda_1x}f_{T_{1,i-1}}(t)\mathrm{d}x \mathrm{d}t\\&\nonumber
=\dfrac{1}{P(E^{\mathrm{B}}_{1,i})}\int_{0}^{\infty} \bigg(-t^2e^{-\lambda_1 t}-\dfrac{t}{\lambda_1}e^{-\lambda_1 t}+\dfrac{t}{\lambda_1}\bigg)f_{T_{1,i-1}}(t)\mathrm{d}t\\&\hspace{5mm}\nonumber
\dfrac{1}{P(E^{\mathrm{B}}_{1,i})}\bigg(-\mathbb{E}[T^2e^{-\lambda_1 T}]-\dfrac{\mathbb{E}[Te^{-\lambda_1 T}]}{\lambda_1}+\dfrac{\mathbb{E}[T]}{\lambda_1}\bigg)\\&\nonumber
\stackrel{(a)}{=}\dfrac{1}{P(E^{\mathrm{B}}_{1,i})}\bigg(-L''_{T}(\lambda_1)+\dfrac{L'_{T}(\lambda_1)}{\lambda_1}+\dfrac{\mathbb{E}[W]+1/\mu}{\lambda_1} \bigg),
\end{align}
where in equality (a) the first and second derivative of the Laplace transform of the PDF of the system time, $L'_{T}$ and $L''_{T}$ at $\lambda_1$, respectively, were obtained using the feature of the Laplace transform that for any function $f(y), y\ge 0$, we have \cite[Sect.~13.5]{rade2013mathematics}
\begin{align}\label{dr001}
L_{y^nf(y)}(a)=(-1)^n\dfrac{\mathrm{d}^n(L_{f(y)}(a))}{{\mathrm{d}a^n}},
\end{align}
and consequently,
\begin{align}\label{dr01}
\mathbb{E}[T^ne^{-a T}]&=(-1)^n\dfrac{\mathrm{d}^n(L_{T}(a))}{{\mathrm{d}a^n}}.
\end{align}

Using Corollary \ref{ghfd0}, the second term $\mathbb{E}[X^2_{1,i}|E^{\mathrm{B}}_{1,i}]$ in \eqref{vbfgh} is calculated as
\begin{align}\label{mn01}
\mathbb{E}[X^2_{1,i}|E^{\mathrm{B}}_{1,i}]&=\int_{0}^{\infty} x^2 f_{X_{1,i}|E^{\mathrm{B}}_{1,i}}(x) \mathrm{d}x = \dfrac{1}{P(E^{\mathrm{B}}_{1,i})}\int_{0}^{\infty} x^2 \lambda_1 e^{-\lambda_1 x}\big(1-F_{T_{1,i-1}}(x)\big)\mathrm{d}x\\&\nonumber=
\dfrac{1}{P(E^{\mathrm{B}}_{1,i})}\bigg(\int_{0}^{\infty} x^2 \lambda_1 e^{-\lambda_1 x}\mathrm{d}x-\lambda_1\int_{0}^{\infty} e^{-\lambda_1 x}\big(x^2F_{T_{1,i-1}}(x)\big)\mathrm{d}x\bigg)\\&\nonumber= \dfrac{1}{P(E^{\mathrm{B}}_{1,i})}\bigg(\dfrac{2}{\lambda^2_1}-\lambda_1 L_{x^2F_{T_1}(x)}(\lambda_1)\bigg).
\end{align}
The Laplace transform $L_{x^2F_{T_1}(x)}(\lambda_1)$ in \eqref{mn01} is given by the following lemma.
\begin{lemm}\label{nvbfh}
$L_{x^2F_{T_1}(x)}(\lambda_1)$ is given as follows:
\begin{align}
L_{x^2F_{T_1}(x)}(a)\bigg|_{a=\lambda_1}=\dfrac{\lambda_1 L''_{T}(\lambda_1)-2 L'_{T}(\lambda_1)}{\lambda_1^2}+\dfrac{2L_{T}(\lambda_1)}{\lambda_1^3}.
\end{align}
\end{lemm}
\begin{proof}
	See Appendix \ref{Lemapp}.
\end{proof}
Thus, applying Lemma \ref{nvbfh}, the conditional expectation  in \eqref{mn01} is given as
\begin{align}\label{1mjn}
\mathbb{E}[X^2_{1,i}|E^{\mathrm{B}}_{1,i}]= \dfrac{1}{P(E^{\mathrm{B}}_{1,i})}\bigg(\dfrac{2}{\lambda^2_1}-L''_{T}(\lambda_1)+\dfrac{2 L'_{T}(\lambda_1)}{\lambda_1}-\dfrac{2L_{T}(\lambda_1)}{\lambda_1^2}\bigg).
\end{align}

Finally, substituting \eqref{102mn} and \eqref{1mjn} in \eqref{vbfgh}, the first conditional expectation $\mathbb{E}[R_{1,i}^{\mathrm{B}}X_{1,i}|E^{\mathrm{B}}_{1,i}]$ in \eqref{mnbv0} is given by
\begin{align}\label{mnb50}
\mathbb{E}[R_{1,i}^{\mathrm{B}}X_{1,i}|E^{\mathrm{B}}_{1,i}]= \dfrac{1}{P(E^{\mathrm{B}}_{1,i})}\bigg(\dfrac{\mathbb{E}[W]+1/\mu}{\lambda_1}-\dfrac{L'_{T}(\lambda_1)}{\lambda_1}+\dfrac{2L_{T}(\lambda_1)}{\lambda_1^2}-\dfrac{2}{\lambda^2_1}\bigg).
\end{align}

\subsubsection{The Second Conditional Expectation in \eqref{mnbv0}}\label{Label_for_subsubsection_2}
Next, we derive the second term $\mathbb{E}[S_{1,i}^{\mathrm{B}}X_{1,i}|E^{\mathrm{B}}_{1,i}]$ in \eqref{mnbv0}. First, let us elaborate the quantity $M_{2,i}^{\mathrm{B}}$ which is an integral part of calculating \eqref{mnbv0}. Recall that $M_{2,i}^{\mathrm{B}}$ is defined as the number of queued packets of source $2$ that must be served before packet $1,i$ according to the FCFS policy under the event ${E_{1,i}^{\mathrm{B}}=\{T_{1,i-1}\ge{X_{1,i}}\}}$. Thus, $M_{2,i}^{\mathrm{B}}$ is equal to the number of arrived (and thus, queued) packets of source $2$ during the (brief) interarrival time $X_{1,i}$. Consequently, we have a Markov chain ${{T_{1,i-1}}\leftrightarrow{X_{1,i}}\leftrightarrow{M_{2,i}^{\mathrm{B}}}}$ conditioned on the event $E_{1,i}^{\mathrm{B}}$, i.e., $M_{2,i}^{\mathrm{B}}$ is independent of $T_{1,i-1}$ given $X_{1,i}$ under the event $E_{1,i}^{\mathrm{B}}$.

Accordingly, the conditional expectation $\mathbb{E}[S_{1,i}^{\mathrm{B}}X_{1,i}|E^{\mathrm{B}}_{1,i}]$ in \eqref{mnbv0} can be expressed as
\begin{align}\label{nmb00}
\mathbb{E}[S_{1,i}^{\mathrm{B}}X_{1,i}|E^{\mathrm{B}}_{1,i}]&=\int_{0}^{\infty}x  \mathbb{E}\bigg[\textstyle\sum_{i'\in \mathcal{M}_{2,i}^{\mathrm{B}}}S_{2,i'}|E^{\mathrm{B}}_{1,i},X_{1,i}=x\bigg] f_{X_{1,i}|E^{\mathrm{B}}_{1,i}}(x)\mathrm{d}x\\
&\nonumber\stackrel{(a)}{=}\dfrac{1}{\mu}\int_{0}^{\infty}x \mathbb{E}\bigg[M_{2,i}^{\mathrm{B}}|X_{1,i}=x\bigg]f_{X_{1,i}|E^{\mathrm{B}}_{1,i}}(x)\mathrm{d}x\\
&\nonumber\stackrel{(b)}{=}\dfrac{ \rho_{2}}{P(E^{\mathrm{B}}_{1,i})}\int_{0}^{\infty}x^2 \lambda_1 e^{-\lambda_1 x}(1-F_{T_{1,i-1}}(x))\mathrm{d}x\\&\nonumber=
\dfrac{ \rho_{2}}{P(E^{\mathrm{B}}_{1,i})}\bigg(\int_{0}^{\infty}x^2 \lambda_1 e^{-\lambda_1 x} \mathrm{d}x-\int_{0}^{\infty}x^2\lambda_1 e^{-\lambda_1 t} F_{T_{1,i-1}}(x)\mathrm{d}x\bigg)\\
&\nonumber\stackrel{(c)}{=}\dfrac{ \rho_{2}}{P(E^{\mathrm{B}}_{1,i})}\bigg(\dfrac{2}{\lambda_1^2}-L''_{T}(\lambda_1)+\dfrac{2 L'_{T}(\lambda_1)}{\lambda_1}-\dfrac{2L_{T}(\lambda_1)}{\lambda_1^2} \bigg),
\end{align}
where equality (a) follows because (i) the service time $S_{2,i'}$ is independent of all other random variables in the system and (ii) by the Markov chain property ${{T_{1,i-1}}\leftrightarrow{X_{1,i}}\leftrightarrow{M_{2,i}^{\mathrm{B}}}}$ conditioned on $E_{1,i}^{\mathrm{B}}$, $M_{2,i}^{\mathrm{B}}$ is independent of $T_{1,i-1}$ given ${X_{1,i}=x}$ under the event ${E_{1,i}^{\mathrm{B}}}$;  equality (b) comes from Corollary \ref{ghfd0} and the fact that ${\mathbb{E}[M_{2,i}^{\mathrm{B}}|X_{1,i}=x]=\lambda_{2}x}$;  equality (c) comes from Lemma \ref{nvbfh}.

\subsubsection{The Third Conditional Expectation in \eqref{mnbv0}}\label{Label_for_subsubsection_3}
The third term ${\mathbb{E}[(S_{1,i}^{\mathrm{L}}\!+\!R^{\mathrm{L}}_{2,i})X_{1,i}|E_{1,i}^{\mathrm{L}}]}$ in \eqref{mnbv0} can be calculated as
\begin{align}\nonumber
&\mathbb{E}[(S_{1,i}^{\mathrm{L}}\!+\!R^{\mathrm{L}}_{2,i})X_{1,i}|E_{1,i}^{\mathrm{L}}]\!\!=\!\!\displaystyle\int_{0}^{\infty}\!\!\int_{0}^{\infty}\!\!\!x\mathbb{E}\!\!\left[\textstyle\sum_{i'\in \mathcal{M}_{2,i}^{\mathrm{L}}}S_{2,i'}|X_{1,i}\!=\!x,T_{1,i-1}\!=\!t,E_{1,i}^{\mathrm{L}}\!\right]\!\!f_{X_{1,i},T_{1,i-1}|E_{1,i}^{\mathrm{L}}}(x,t)\mathrm{d}x\mathrm{d}t\\&\label{mg1main}+\displaystyle\int_{0}^{\infty}\int_{0}^{\infty}x\mathbb{E}\left[R^{\mathrm{L}}_{2,i}|X_{1,i}=x,T_{1,i-1}=t,E_{1,i}^{\mathrm{L}}\right]f_{X_{1,i},T_{1,i-1}|E_{1,i}^{\mathrm{L}}}(x,t)\mathrm{d}x\mathrm{d}t,
\end{align}
where the first term on the right hand side can be calculated as
\begin{align}\nonumber
&\displaystyle\int_{0}^{\infty}\!\!\int_{0}^{\infty}\!\!\!x\mathbb{E}\!\!\left[\textstyle\sum_{i'\in \mathcal{M}_{2,i}^{\mathrm{L}}}S_{2,i'}|X_{1,i}\!=\!x,T_{1,i-1}\!=\!t,E_{1,i}^{\mathrm{L}}\!\right]\!\!f_{X_{1,i},T_{1,i-1}|E_{1,i}^{\mathrm{L}}}(x,t)\mathrm{d}x\mathrm{d}t\\&\label{mg1main0}\overset{(a)}{=}\dfrac{1}{\mu}\displaystyle\int_{0}^{\infty}\int_{0}^{\infty}x\mathbb{E}\left[M^{\mathrm{L}}_{2,i}|X_{1,i}=x,T_{1,i-1}=t,E_{1,i}^{\mathrm{L}}\right]f_{X_{1,i},T_{1,i-1}|E_{1,i}^{\mathrm{L}}}(x,t)\mathrm{d}x\mathrm{d}t\\
&\nonumber=\dfrac{1}{\mu}\displaystyle\int_{0}^{\infty}\int_{0}^{\infty}x\textstyle\sum_{m=0}^{\infty}m\mathrm{Pr}[M_{2,i}^{\mathrm{L}}=m|X_{1,i}=x,T_{1,i-1}=t,E_{1,i}^{\mathrm{L}}]f_{X_{1,i},T_{1,i-1}|E_{1,i}^{\mathrm{L}}}(x,t)\mathrm{d}x\mathrm{d}t,
\end{align}
where equality (a) follows because (i) the service time $S_{2,i'}$ is independent of all other  random variables in the system and (ii) the expectation of a sum of  random number $N$ independent and identically distributed  random variables ${Y_n, n=1, \dots, N},$ is equal to the expectation of the random number $\mathbb{E}[N]$ times the expectation of a random variable $\mathbb{E}[Y_n]$, i.e., ${\mathbb{E}[\sum_{n=1}^{N}Y_n]=\mathbb{E}[N]E[Y_n]}$  \cite[Sect.~11.2]{ross2014introduction}.

\textbf{Remark 1.}
The  second term on the right hand side of \eqref{mg1main} and the final expression in \eqref{mg1main0} reveal two critical issues in deriving the third conditional expectation term of \eqref{mnbv0}. The second term on the right hand side of \eqref{mg1main} contains the possible residual service time of the  packet of source 2 that is under service at the arrival instant of packet ${1,i}$, $R^{\mathrm{L}}_{2,i}$, which cannot be further simplified. 
%
%
%
%
{In the final expression of \eqref{mg1main0}, we need to calculate the time-dependent probability of the number of packets in an M/G/1 queue with source 2 packet arrivals, i.e., $\mathrm{Pr}[M_{2,i}^{\mathrm{L}}=m|X_{1,i}=x,T_{1,i-1}=t,E_{1,i}^{\mathrm{L}}]$. Computing this time-dependent probability in an M/G/1 queueing model is complicated and needs the transient analysis of an	M/G/1 queueing model. While characterizations of the transient behavior of an M/G/1 queue are investigated in some works such as \cite{Takacs1995}, to the best of our knowledge,  such a time-dependent probability has not been derived before in closed form so that it could be used in deriving the conditional expectation in \eqref{mg1main0}.}

%
%
%
%
%
%
%
%
%
%
%
%
%
 Fortunately, these difficulties can be overcome when the service time is exponential, i.e., in an M/M/1 queueing model. Thus, we proceed as follows. We derive an \emph{exact} expression of the average AoI in a multi-source M/M/1 queueing model in Section \ref{Exact Expression for the Average AoI in a Multi-Source M/M/1 Queueing Model}. In Section \ref{An Approximation for the Average AoI in a Multi-Source M/G/1 Queueing Model}, we propose three approximations for \eqref{mg1main} and derive three \emph{approximate} expressions for the average AoI in a multi-source M/G/1 queueing model.

%

\section{Exact Expression for the Average AoI in a Multi-Source M/M/1 Queueing Model}\label{Exact Expression for the Average AoI in a Multi-Source M/M/1 Queueing Model}
In this section, we derive the \emph{exact} expression of the average AoI in \eqref{oointr00103} for a multi-source M/M/1 queueing model that was stated in Theorem \ref{exactmm1theorem} in Section II.B. Recall that in Section \ref{AoI_Multi}, we already derived general expressions (for an M/G/1 case) for the key terms needed to describe the average AoI, i.e., the three conditional expectation terms of \eqref{mnbv0}, which are given in \eqref{mnb50}, \eqref{nmb00}, and \eqref{mg1main}, respectively. Next, we specify these three terms to the case with exponentially distributed service time. We start by deriving an exact expression for the most challenging term, i.e., the third term \eqref{mg1main}, followed by the calculation of \eqref{mnb50} and \eqref{nmb00}.

	 Focus now on \eqref{mg1main}. Due to the memoryless property of the exponentially distributed service time,  the possible residual service time of the packet of source $2$ that is under service at the arrival instant of packet ${1,i}$ for event $E_{1,i}^{\mathrm{L}}$ is also exponentially distributed; thus, the waiting time is the sum of $ \hat{M}_{2,i}^{\mathrm{L}} $ exponentially distributed random variables, where $ \hat{M}_{2,i}^{\mathrm{L}} $ is the total number of source 2 packets in the \emph{system} (either in the queue or under service) at the arrival instant of packet $1,i$ conditioned
	 on the event $E_{1,i}^{\mathrm{L}}$  \cite[p.~168]{Kleinrock121}. Therefore, the waiting time in \eqref{mg1main} can be expressed as
\begin{equation}\label{WL_MM1}
W_{1,i}=S_{1,i}^{\mathrm{L}}+R^{\mathrm{L}}_{2,i}=\textstyle\sum_{i'\in \mathcal{\hat{M}}_{2,i}^{\mathrm{L}}}S_{2,i'},
\end{equation}
where $\mathcal{\hat{M}}_{2,i}^{\mathrm{L}}$ is the set of indices of packets of source $2$ that are in the system  at the arrival instant of packet $1,i$ for event $E_{1,i}^{\mathrm{L}}$, with $|\mathcal{\hat{M}}_{2,i}^{\mathrm{L}}|=\hat{M}_{2,i}^{\mathrm{L}}$. 

By \eqref{WL_MM1}, $\mathbb{E}[{W}_{1,i}X_{1,i}|E_{1,i}^{\mathrm{L}}] $ (cf. \eqref{mg1main})
	can be calculated as
\begin{align}\label{bnbm}
&\mathbb{E}[{W}_{1,i}X_{1,i}|E_{1,i}^{\mathrm{L}}]=\displaystyle\int_{0}^{\infty}\int_{0}^{\infty}x\mathbb{E}\left[\textstyle\sum_{i'\in \mathcal{\hat{M}}_{2,i}^{\mathrm{L}}}S_{2,i'}|X_{1,i}=x,T_{1,i-1}=t,E_{1,i}^{\mathrm{L}}\right]f_{X_{1,i}T_{1,i-1}|E_{1,i}^{\mathrm{L}}}(x,t)\mathrm{d}x\mathrm{d}t\nonumber\\
&{=}\dfrac{1}{\mu}\displaystyle\int_{0}^{\infty}\int_{0}^{\infty}x\mathbb{E}\left[\hat{M}_{2,i}^{\mathrm{L}}|X_{1,i}=x,T_{1,i-1}=t,E_{1,i}^{\mathrm{L}}\right]f_{X_{1,i}T_{1,i-1}|E_{1,i}^{\mathrm{L}}}(x,t)\mathrm{d}x\mathrm{d}t\\\nonumber
&{=}\dfrac{1}{\mu}\displaystyle\int_{0}^{\infty}\int_{0}^{\infty}x\textstyle\sum_{m=0}^{\infty}m\mathrm{Pr}[\hat{M}_{2,i}^{\mathrm{L}}=m|X_{1,i}=x,T_{1,i-1}=t,E_{1,i}^{\mathrm{L}}]f_{X_{1,i}T_{1,i-1}|E_{1,i}^{\mathrm{L}}}(x,t)\mathrm{d}x\mathrm{d}t.
\end{align}

Next, we calculate $\mathrm{Pr}[\hat{M}_{2,i}^{\mathrm{L}}=m|X_{1,i}=x,T_{1,i-1}=t,E_{1,i}^{\mathrm{L}}]$ in \eqref{bnbm} by introducing an \emph{auxiliary} random variable $J_{2,i}^{\mathrm{L}}$ that represents the number of source 2 packets in the system at the departure instant of packet $1,i-1$ for event $E_{1,i}^{\mathrm{L}}$ (see Fig. \ref{AoI021}(b)).
	Using the law of total expectation, $\mathrm{Pr}[\hat{M}_{2,i}^{\mathrm{L}}=m|X_{1,i}=x,T_{1,i-1}=t,E_{1,i}^{\mathrm{L}}]$ in \eqref{bnbm} is written as
\begin{align}\label{mnvb01}
&\mathrm{Pr}[\hat{M}_{2,i}^{\mathrm{L}}=m|X_{1,i}=x,T_{1,i-1}=t,E_{1,i}^{\mathrm{L}}]\\
&\nonumber =\textstyle\sum_{j=0}^{\infty}\mathrm{Pr}[\hat{M}_{2,i}^{\mathrm{L}}=m|J_{2,i}^{\mathrm{L}}=j,X_{1,i}=x,T_{1,i-1}=t,E_{1,i}^{\mathrm{L}}]\mathrm{Pr}[J_{2,i}^{\mathrm{L}}=j|X_{1,i}=x,T_{1,i-1}=t,E_{1,i}^{\mathrm{L}}],
\end{align}
where
\begin{align}\label{PrJ}
\mathrm{Pr}[J_{2,i}^{\mathrm{L}}=j|X_{1,i}=x,T_{1,i-1}=t,E_{1,i}^{\mathrm{L}}]&\overset{(a)}{=}\mathrm{Pr}[J_{2,i}^{\mathrm{L}}=j|T_{1,i-1}=t,E_{1,i}^{\mathrm{L}}]\overset{(b)}{=}\displaystyle{e^{-\lambda_{2}t}\frac{{(\lambda_{2}t)}^{j}}{j!}},
\end{align}
where equality $(a)$ follows because $J_{2,i}^{\mathrm{L}}$ is conditionally  independent of $X_{1,i}$ given $T_{1,i-1}$ and $E_{1,i}^{\mathrm{L}}$; equality $(b)$ follows because (i) under the long event $E_{1,i}^{\mathrm{L}}$, all $J_{2,i}^{\mathrm{L}}$ source 2 packets that are in the system at the departure instant of packet $ 1,i-1 $ must have arrived during the system time $T_{1,i-1}$ (see Fig. \ref{AoI021}(b)), and (ii) the probability of having $ j $ Poisson arrivals of rate $\lambda_2$ during the time interval $ T_{1,i-1}=t $ is $e^{-\lambda_{2}t}\frac{{(\lambda_{2}t)}^{j}}{j!}$ \cite[Eq.~(2.119)]{Kleinrock121}.

%
%
%
%
Focus now on term $\mathrm{Pr}[\hat{M}_{2,i}^{\mathrm{L}}=m|J_{2,i}^{\mathrm{L}}=j,X_{1,i}=x,T_{1,i-1}=t,E_{1,i}^{\mathrm{L}}]$ in \eqref{mnvb01}. Note that during the time interval between the departure of packet $ 1,i-1 $ and the arrival of packet $ 1,i $ (i.e., $ (t'_{1,i-1}, t_{1,i}) $ in Fig. 2) the queue receives packets only from source 2 and, therefore the system behaves as a single-source M/M/1 queue. Thus, $\mathrm{Pr}[\hat{M}_{2,i}^{\mathrm{L}}=m|J_{2,i}^{\mathrm{L}}=j,X_{1,i}=x,T_{1,i-1}=t,E_{1,i}^{\mathrm{L}}]$ in \eqref{mnvb01} represents the probability that a single-source M/M/1 queueing system with arrival rate $ \lambda_2 $ and which initially holds  $ j $ packets (either in the queue or under service) ends up holding $ m $ packets after $ \tau=x-t $ seconds. We denote this probability compactly by $\bar{P}_{m|j}(\tau)$ and it is given by the transient analysis of an M/M/1 queueing system as \cite[Eq.~(6)]{TCoMstabelqueue01},\cite[Eq.~(2.163)]{Kleinrock121}
\begin{align}\label{Pkj}
&\bar{P}_{m|j}(\tau)=e^{-(\lambda_2+\mu)\tau}\left[\rho_2^{(m-1)/2}I_{m-1}(2\sqrt{\mu\lambda_2}\tau)+\rho_2^{(m-j-1)/2}I_{m+j+1}(2\sqrt{\mu\lambda_2}\tau)\right]\\&\nonumber+\rho_2^m(1-\rho_2)\big(1-Q_{m+j+2}(\sqrt{2\lambda_2\tau},\sqrt{2\mu\tau})\big).
\end{align}
where  $I_k(\cdot)$ represents the modified Bessel function of the first kind of order $k$, and $Q_k(a,b)$ is the generalized Q-function.

Substituting \eqref{mnvb01}, \eqref{PrJ}, and \eqref{Pkj} into \eqref{bnbm}, we have
\begin{align}
&\nonumber\mathbb{E}[{W}_{1,i}X_{1,i}|E_{1,i}^{\mathrm{L}}]{=}\dfrac{1}{\mu}\displaystyle\int_{0}^{\infty}\int_{0}^{\infty}x\textstyle\sum_{m=0}^{\infty}\textstyle\sum_{j=0}^{\infty}m\bar{P}_{m|j}(x-t)e^{-\lambda_{2}t}\displaystyle\frac{{(\lambda_{2}t)}^{j}}{j!}f_{X_{1,i}T_{1,i-1}|E_{1,i}^{\mathrm{L}}}(x,t)\mathrm{d}x\mathrm{d}t\\
&\label{01mnk}\overset{(a)}{=}\dfrac{\lambda_1(1-\rho)}{P(E_{1,i}^{\mathrm{L}})}\!\!\!\displaystyle\int_{0}^{\infty}\!\!\!\int_{0}^{\infty}\!\!\!\!(t\!+\!\tau){e^{-\mu(t+\rho_1\tau)}}\bigg(\!\textstyle\sum_{m=0}^{\infty}\textstyle\sum_{j=0}^{\infty}m\bar{P}_{m|j}(\tau)\displaystyle\frac{{(\lambda_{2}t)}^{j}}{j!}\bigg)\mathrm{d}\tau\mathrm{d}t\\
&\nonumber\triangleq\dfrac{\lambda_1(1-\rho)}{P(E_{1,i}^{\mathrm{L}})}\Psi(\mu,\rho_1,\lambda_2),
\end{align}
where $(a)$ follows from the substitution ${\tau=x-t}$ and Lemma \ref{gh0fd00} (below) which derives the conditional PDF $f_{X_{1,i},T_{1,i-1}|E^{\mathrm{L}}_{1,i}}(x,t)$. Note that the double integral in $\Psi(\mu,\rho_1,\lambda_2)$ needs to be in general numerically calculated.

\begin{lemm}\label{gh0fd00}
The conditional PDF  $f_{X_{1,i},T_{1,i-1}|E^{\mathrm{L}}_{1,i}}(x,t)$ is given by
\begin{align}\label{mbnjhf}
f_{X_{1,i},T_{1,i-1}|E^{\mathrm{L}}_{1,i}}(x,t) =\begin{cases}0&x<t\\\dfrac{\lambda_1e^{-\lambda_1 x}f_{T_{1,i-1}}(t) }{P(E^{\mathrm{L}}_{1,i})}&x\ge t.\end{cases}
\end{align}	
\end{lemm}
\begin{proof}The proof of Lemma \ref{gh0fd00} follows from the similar steps as used for Lemma \ref{gh0fd0}.
\end{proof}

By substituting the probabilities $P(E^{\mathrm{B}}_{1,i})$ and $P(E^{\mathrm{L}}_{1,i})$ given by Lemma \ref{lemmii1} and the three derived conditional expectation terms  \eqref{mnb50},  \eqref{nmb00}, and \eqref{01mnk} into \eqref{mnbv0}, $ \mathbb{E}[X_{1,i}W_{1,i}] $ can be expressed as
\begin{align}\label{mn01nm}
\mathbb{E}[X_{1,i}W_{1,i}]=& \dfrac{\mathbb{E}[W]}{\lambda_1}+\lambda_1(1-\rho)\Psi(\mu,\rho_1,\lambda_2)+\dfrac{2(\rho_2-1)}{\lambda^2_1}+\dfrac{1}{\lambda_1\mu}+\dfrac{2(1-\rho_2)}{\lambda^2_1}L_{T}(\lambda_1)\\
&\nonumber+\dfrac{2\rho_2-1}{\lambda_1}L'_{T}(\lambda_1)-\rho_2L''_{T}(\lambda_1). 
\end{align}
Finally, by substituting \eqref{mn01nm} and \eqref{mbn0} into \eqref{oointr00103}, the average AoI of source 1 for a multi-source M/M/1 queueing model is expressed as:
\begin{align}\label{mm1a}
\Delta_1=& \mathbb{E}[W]+\lambda_1^2(1-\rho)\Psi(\mu,\rho_1,\lambda_2)+\dfrac{2}{\mu}\left(\dfrac{\lambda_2}{\lambda_1}+1\right)-1/\lambda_1+\dfrac{2(1-\rho_2)}{\lambda_1}L_{T}(\lambda_1)\\
&\nonumber+(2\rho_2-1)L'_{T}(\lambda_1)-\lambda_1\rho_2L''_{T}(\lambda_1),
\end{align}
where the average waiting time of each packet in the system, $ \mathbb{E}[W]$,  is given as  \cite[Sect.~3]{bertsekas1992data}
\begin{align}\label{queue000}
\mathbb{E}[W]=\dfrac{\mathbb{E}[S^2]\lambda}{2(1-\rho)},
\end{align}
where ${\mathbb{E}[S^2]=2/\mu^2}$ is the second moment of the service time, $L_{T}(\lambda_1)$ is a function of the Laplace transform of the PDF of the service time given by \cite[Sect.~5.1.2]{daigle2005basic}
\begin{align}\label{000bn0v001b0012}
L_{T}(\lambda_1)=\dfrac{\big(1-\rho\big)\lambda_1L_{S}(\lambda_1)}{\lambda_1-\lambda\big(1-L_{S}(\lambda_1)\big)},
\end{align}
and $L'_{T}(\lambda_1)$ and $L''_{T}(\lambda_1)$ are the first and second derivative of $L_{T}(\cdot)$ at $\lambda_1$, respectively, as
\begin{align}\label{rgn01}
L'_{T}(\lambda_1)&=\dfrac{\mathrm{d}(L_{T}(a))}{\mathrm{d}a}\bigg|_{a=\lambda_1}\!\!\!={(1-\rho)}\dfrac{{\lambda} L^2_{S}(\lambda_1)+\big(\lambda^2_1-\lambda_1\lambda\big)L'_{S}(\lambda_1)-\lambda L_{S}(\lambda_1)}{\big(\lambda_1-\lambda\big(1-L_{S}(\lambda_1)\big)\big)^2},\\\nonumber
L''_{T}(\lambda_1)&=\dfrac{\mathrm{d}^2(L_{T}(a))}{\mathrm{d}a^2}\bigg|_{a=\lambda_1}\nonumber\!\!\!={(1-\rho)}\bigg(\dfrac{{\lambda} L''_{S}(\lambda_1)\big(\lambda^2_1-\lambda_1\lambda\big)+2L'_{S}(\lambda_1)\big(\lambda_1-\lambda+\lambda L_{S}(\lambda_1)\big)}{\big(\lambda_1-\lambda\big(1-L_{S}(\lambda_1)\big)\big)^2}\\&\nonumber
-\dfrac{2({\lambda} L^2_{S}(\lambda_1)+\big(\lambda^2_1-\lambda_1\lambda\big)L'_{S}(\lambda_1)-\lambda L_{S}(\lambda_1))(1+\lambda L'_{S}(\lambda_1))}{\left(\lambda_1-\lambda\left(1-L_{S}(\lambda_1)\right)\right)^3}\bigg),
\end{align}
where $L'_{S}(\lambda_1)$ and $L''_{S}(\lambda_1)$  for  the exponential service time are computed according to \eqref{dr001} as
\begin{align}\label{02158lk0}
L_{S}(\lambda_1)&=\int_{0}^{\infty}\mu e^{-(\mu+\lambda_1) s}\mathrm{d}s=\dfrac{\mu}{\mu+\lambda_1},\quad
L'_{S}(\lambda_1)=-\int_{0}^{\infty}s\mu e^{-(\mu+\lambda_1) s}\mathrm{d}s=-\dfrac{\mu}{(\mu+\lambda_1)^2},\\\nonumber
L''_{S}(\lambda_1)&=\int_{0}^{\infty}s^2\mu e^{-(\mu+\lambda_1) s}\mathrm{d}s=\dfrac{2\mu}{(\mu+\lambda_1)^3}.
\end{align}
%
%
%
Finally, by substituting $\mathbb{E}[W]$, $L_{T}(\lambda_1)$, $L'_{T}(\lambda_1)$, and $L''_{T}(\lambda_1)$ into \eqref{mm1a} we get the result in Theorem \ref{exactmm1theorem} in Section \ref{Summary of the Main Results}, i.e., the average AoI of source 1 for a multi-source M/M/1 queueing model is given as
\begin{align}\label{mm10a0}
\Delta_1&\!= \!\lambda^2_1(1-\rho)\Psi(\mu,\rho_1,\lambda_2)\!+\!\dfrac{1 }{\mu}\bigg(\dfrac{1 }{\rho_1}\!+\!
\dfrac{\rho}{1-\rho}\!+\!\dfrac{(2\rho_2-1)(\rho - 1)}{(1-\rho_2)^2}+\dfrac{2\rho_1\rho_2(\rho - 1)}{(1-\rho_2)^3}\bigg).
\end{align}

\textbf{Remark 2.}
It is worth noting that \eqref{mm10a0} does not coincide with the prior result \cite[Theorem.~1]{8469047} and \cite[Eq.~(16)]{6284003}. The dissimilarity is explained in the following. The authors of \cite{6284003,8469047}  considered a similar two-source FCFS M/M/1 queueing model, with the aim of deriving a  closed-form expression for the  average AoI of source 1 ($\Delta_1$).
Let us focus on \cite[Eq.~(33)]{6284003} where the authors compute a conditional expectation equivalent to our  $\mathbb{E}[W_{1,i}X_{1,i}|E_{1,i}^{\mathrm{L}}]$ given by \eqref{01mnk}, which by \eqref{WL_MM1} can be expressed as
\begin{align}\label{0022mn}
\mathbb{E}[W_{1,i}X_{1,i}|E^{\mathrm{L}}_{1,i}]=\mathbb{E}\Big[\textstyle\sum_{i'\in \mathcal{\hat{M}}_{2,i}^{\mathrm{L}}}S_{2,i'}X_{1,i}|E^{\mathrm{L}}_{1,i}\Big].
\end{align}
The  authors of \cite{6284003} tacitly assumed conditional independency between $\textstyle\sum_{i'\in \mathcal{\hat{M}}_{2,i}^{\mathrm{L}}}S_{2,i'}$ and $X_{1,i}$ under the event ${E^{\mathrm{L}}_{1,i}=\{T_{1,i-1}<X_{1,i}\}}$, and calculated \eqref{0022mn} as a multiplication of two  expectations as
\begin{align}\label{00022mn}
\mathbb{E}[W_{1,i}X_{1,i}|E^{\mathrm{L}}_{1,i}]=
\mathbb{E}\Big[\textstyle\sum_{i'\in \mathcal{\hat{M}}_{2,i}^{\mathrm{L}}}S_{2,i'}|T_{1,i-1}<X_{1,i}\Big]\mathbb{E}\Big[X_{1,i}|T_{1,i-1}<X_{1,i}\Big].
\end{align}
The critical point is that even if $X_{1,i}$ is independent of $T_{1,i-1}$, they become dependent when conditioned on the event ${E_{1,i}^{\mathrm{L}}=\{T_{1,i-1}<{X_{1,i}}\}}$, as in \eqref{0022mn}. This conditional dependency is violated by the separation of the expectations in \eqref{00022mn} because the quantity $\hat{M}_{2,i}^{\mathrm{L}}$ in general depends on both $T_{1,i-1}$ and ${X_{1,i}}$, and, thus, the multiplicative quantities $\textstyle\sum_{i'\in \mathcal{\hat{M}}_{2,i}^{\mathrm{L}}}S_{2,i'}$ and $X_{1,i}$ are dependent under the event $E_{1,i}^{\mathrm{L}}$. Note that we incorporate this conditional dependency in calculating $\mathbb{E}[W_{1,i}X_{1,i}|E^{\mathrm{L}}_{1,i}]$ by using the conditional joint PDF $f_{X_{1,i},T_{1,i-1}|E^{\mathrm{L}}_{1,i}}(x,t)$.

\textbf{Remark 3.}
It is worth to note that \eqref{mm10a0} neither coincides with our prior result \cite[Eq.~(25)]{ourITW01}. The dissimilarity comes from the fact that in \cite{ourITW01},  we  wrongly used steady-state properties of a queueing system in calculating ${\mathbb{E}\left[\hat{M}_{2,i}^{\mathrm{L}}|X_{1,i}=x,T_{1,i-1}=t,E_{1,i}^{\mathrm{L}}\right]}$ in \eqref{bnbm}.
{\section{ Approximate Expressions for the Average AoI in a Multi-Source M/G/1 Queueing Model}\label{An Approximation for the Average AoI in a Multi-Source M/G/1 Queueing Model}
In this section, we derive the three \emph{approximate} expressions of the average AoI in \eqref{oointr00103} for a multi-source M/G/1 queueing model that were presented in Section II.B. Recall that the exact expressions for the first and second conditional expectation  terms of \eqref{mnbv0} are given by \eqref{mnb50} and \eqref{nmb00}, respectively. 
 From \eqref{mg1main} and  \eqref{mg1main0}, the third conditional expectation is given as
 \begin{align}\nonumber
 &\mathbb{E}[(S_{1,i}^{\mathrm{L}}\!+\!R^{\mathrm{L}}_{2,i})X_{1,i}|E_{1,i}^{\mathrm{L}}]\!\!=\!\!\dfrac{1}{\mu}\displaystyle\int_{0}^{\infty}\int_{0}^{\infty}x\mathbb{E}\left[M^{\mathrm{L}}_{2,i}|X_{1,i}=x,T_{1,i-1}=t,E_{1,i}^{\mathrm{L}}\right]f_{X_{1,i},T_{1,i-1}|E_{1,i}^{\mathrm{L}}}(x,t)\mathrm{d}x\mathrm{d}t\\&\label{mg1ma0in}+\displaystyle\int_{0}^{\infty}\int_{0}^{\infty}x\mathbb{E}\left[R^{\mathrm{L}}_{2,i}|X_{1,i}=x,T_{1,i-1}=t,E_{1,i}^{\mathrm{L}}\right]f_{X_{1,i},T_{1,i-1}|E_{1,i}^{\mathrm{L}}}(x,t)\mathrm{d}x\mathrm{d}t.
 \end{align}
Next, we propose three approximate calculations for the third conditional expectation term of \eqref{mnbv0},  given by \eqref{mg1ma0in},  differing in the way we approximate the terms  $\mathbb{E}\left[M^{\mathrm{L}}_{2,i}|X_{1,i}=x,T_{1,i-1}=t,E_{1,i}^{\mathrm{L}}\right]$ and $\mathbb{E}\left[R^{\mathrm{L}}_{2,i}|X_{1,i}=x,T_{1,i-1}=t,E_{1,i}^{\mathrm{L}}\right]$. }

 {\textbf{Approximation 1:}
 First, we neglect the possible residual service time of source $2$ packet that is  under service at the arrival instant of packet ${1,i}$.
 Second, we assume that the average number of packets of source $2$ that must be served before packet $1,i$ is equal to the average number of packets of source $2$ that are queued during the system time of packet $1,i-1$ ($T_{1,i-1}$). Thus, we assume  $\mathbb{E}\left[M^{\mathrm{L}}_{2,i}|X_{1,i}=x,T_{1,i-1}=t,E_{1,i}^{\mathrm{L}}\right]= \mathbb{E}\left[J_{2,i}^{\mathrm{L}}|X_{1,i}=x,T_{1,i-1}=t,E_{1,i}^{\mathrm{L}}\right]$, where, as defined previously, the random variable $J_{2,i}^{\mathrm{L}}$ represents the number of source $2$ packets in the system  at the departure instant of packet $1,i-1$ for the long event $E_{1,i}^{\mathrm{L}}$. With the simplifications above, \eqref{mg1ma0in} can be approximated as
 \begin{align}
 &\nonumber\mathbb{E}[(S_{1,i}^{\mathrm{L}}\!+\!R^{\mathrm{L}}_{2,i})X_{1,i}|E_{1,i}^{\mathrm{L}}]
 {\approx} \dfrac{1}{\mu}\int_{0}^{\infty}\int_{0}^{\infty}x\mathbb{E}\left[J_{2,i}^{\mathrm{L}}|X_{1,i}=x,T_{1,i-1}=t,E_{1,i}^{\mathrm{L}}\right]f_{X_{1,i},T_{1,i-1}|E^{\mathrm{L}}_{1,i}}(x,t)\mathrm{d}x\mathrm{d}t
 \\&\label{mnvbh00}\stackrel{(a)}{=}
 \rho_{2}\int_{0}^{\infty}\int_{0}^{\infty}tx f_{X_{1,i},T_{1,i-1}|E^{\mathrm{L}}_{1,i}}(x,t)\mathrm{d}x\mathrm{d}t\\
 &\nonumber\stackrel{(b)}{=}\dfrac{\rho_{2}}{P(E^{\mathrm{L}}_{1,i})}\int_{0}^{\infty}\int_{t}^{\infty}xt \lambda_1 e^{-\lambda_1 x}f_{T_{1,i-1}}(t)\mathrm{d}x\mathrm{d}t\\
 &\nonumber=\dfrac{\rho_{2}}{P(E^{\mathrm{L}}_{1,i})}\int_{0}^{\infty}\bigg(t^2e^{-\lambda_1t}f_{T_{1,i-1}}(t)+\dfrac{te^{-\lambda_1t}}{\lambda_1}f_{T_{1,i-1}}(t)\bigg)\mathrm{d}t\\&\nonumber\stackrel{(c)}{=} \dfrac{\rho_{2}}{P(E^{\mathrm{L}}_{1,i})}\bigg(L''_{T}(\lambda_1)-\dfrac{L'_{T}(\lambda_1)}{\lambda_1}\bigg),
 \end{align}
 where  $ (a) $ comes from the fact that $\mathbb{E}\left[J_{2,i}^{\mathrm{L}}|X_{1,i}=x,T_{1,i-1}=t,E_{1,i}^{\mathrm{L}}\right]=\lambda_{2}t$,  $(b)$ follows from Lemma \ref{gh0fd00}, and $ (c) $ follows from \eqref{dr01}.}

 {By substituting the probabilities $P(E^{\mathrm{B}}_{1,i})$ and $P(E^{\mathrm{L}}_{1,i})$ given by Lemma \ref{lemmii1} and the three derived conditional expectation terms  \eqref{mnb50},  \eqref{nmb00}, and \eqref{mnvbh00} into \eqref{mnbv0}, an approximation for $ \mathbb{E}[X_{1,i}W_{1,i}] $ can be expressed as
 \begin{align}\label{mn01nm0}
 \mathbb{E}[X_{1,i}W_{1,i}] \approx\dfrac{1}{\lambda_1}\left(\mathbb{E}[W]+\dfrac{1}{\mu}+\dfrac{2(\rho_{2}-1)}{\lambda_1}+ \dfrac{2(1-\rho_{2})}{\lambda_1}L_{T}(\lambda_1)+(\rho_{2}-1)L'_{T}(\lambda_1)\right).
 \end{align}
 By substituting \eqref{mn01nm0} and \eqref{mbn0} into \eqref{oointr00103}, an approximation for the average AoI of source $1$ in a multi-source M/G/1 queueing model is given as
\begin{align}\label{main542}
\Delta^{\text{app}_1}_{1}&\approx\mathbb{E}[W]+\dfrac{2}{\mu}+\dfrac{2\rho_{2}-1}{\lambda_1}+ \dfrac{2(1-\rho_{2})}{\lambda_1}L_{T}(\lambda_1)+(\rho_{2}-1)L'_{T}(\lambda_1),
\end{align}
 where the quantities $\mathbb{E}[W]$, $L_{T}(\lambda_1)$, and $L'_{T}(\lambda_1)$ are calculated by \eqref{queue000} -- \eqref{02158lk0} for a specific service time distribution.}

 {\textbf{Approximation 2:}
 First, we assume that the average residual service time of source $2$ packet that is  under service at the arrival instant of packet ${1,i}$ is equal to the average service time of one packet in the system. Thus, we assume that
 ${\mathbb{E}\left[R^{\mathrm{L}}_{2,i}|X_{1,i}=x,T_{1,i-1}=t,E_{1,i}^{\mathrm{L}}\right]=\dfrac{1}{\mu}}$.
 Second, for the term $\mathbb{E}\left[M^{\mathrm{L}}_{2,i}|X_{1,i}=x,T_{1,i-1}=t,E_{1,i}^{\mathrm{L}}\right]$ we use the same approximation as we  used for Approximation 1, i.e.,  $\mathbb{E}\left[M^{\mathrm{L}}_{2,i}|X_{1,i}=x,T_{1,i-1}=t,E_{1,i}^{\mathrm{L}}\right]= \mathbb{E}\left[J_{2,i}^{\mathrm{L}}|X_{1,i}=x,T_{1,i-1}=t,E_{1,i}^{\mathrm{L}}\right]$. Based on these simplifications, \eqref{mg1ma0in} can be approximated as
 \begin{align}
 &\nonumber\mathbb{E}[(S_{1,i}^{\mathrm{L}}\!+\!R^{\mathrm{L}}_{2,i})X_{1,i}|E_{1,i}^{\mathrm{L}}]
 {\approx} \dfrac{1}{\mu}\int_{0}^{\infty}\int_{0}^{\infty}x\mathbb{E}\left[J_{2,i}^{\mathrm{L}}|X_{1,i}=x,T_{1,i-1}=t,E_{1,i}^{\mathrm{L}}\right]f_{X_{1,i},T_{1,i-1}|E^{\mathrm{L}}_{1,i}}(x,t)\mathrm{d}x\mathrm{d}t
 \\&\label{mnvbh000}+\dfrac{1}{\mu}\displaystyle\int_{0}^{\infty}\int_{0}^{\infty}xf_{X_{1,i},T_{1,i-1}|E_{1,i}^{\mathrm{L}}}(x,t)\mathrm{d}x\mathrm{d}t\\&\nonumber{=}
 \rho_{2}\int_{0}^{\infty}\int_{0}^{\infty}tx f_{X_{1,i},T_{1,i-1}|E^{\mathrm{L}}_{1,i}}(x,t)\mathrm{d}x\mathrm{d}t+\dfrac{1}{\mu}\displaystyle\int_{0}^{\infty}\int_{0}^{\infty}xf_{X_{1,i},T_{1,i-1}|E_{1,i}^{\mathrm{L}}}(x,t)\mathrm{d}x\mathrm{d}t\\
 &\nonumber\stackrel{}{=} \dfrac{1}{P(E^{\mathrm{L}}_{1,i})}\bigg(\rho_{2}L''_{T}(\lambda_1)-\left(\dfrac{\rho_2}{\lambda_1}+\dfrac{1}{\mu}\right)L'_{T}(\lambda_1)+\dfrac{L_{T}(\lambda_1)}{\mu\lambda_1}\bigg).
 \end{align}
Using \eqref{mnvbh000} and following the steps used to derive \eqref{mn01nm0}, an approximation for $ \mathbb{E}[X_{1,i}W_{1,i}] $ under Approximation 2 is given  as
\begin{align}\label{mn01nm02}
\mathbb{E}[X_{1,i}W_{1,i}] \approx\dfrac{1}{\lambda_1}\left(\mathbb{E}[W]\!+\!\dfrac{1}{\mu}\!+\!\dfrac{2(\rho_{2}\!-\!1)}{\lambda_1}\!+\! \left(\dfrac{1}{\mu}\!+\!\dfrac{2(1\!-\!\rho_{2})}{\lambda_1}\right)L_{T}(\lambda_1)\!+\!\left(\rho_{2}\!-\!1\!-\!\dfrac{\lambda_1}{\mu}\right)L'_{T}(\lambda_1)\right).
\end{align}
By substituting \eqref{mn01nm02} and \eqref{mbn0} into \eqref{oointr00103}, an approximation for the average AoI of source $1$ in a multi-source M/G/1 queueing model is given as
 \begin{align}\label{main543}
 \Delta^{\text{app}_2}_{1}&\approx\mathbb{E}[W]\!+\!\dfrac{2}{\mu}\!+\!\dfrac{2\rho_{2}\!-\!1}{\lambda_1}\!+\! \left(\dfrac{1}{\mu}\!+\!\dfrac{2(1\!-\!\rho_{2})}{\lambda_1}\right)L_{T}(\lambda_1)\!+\!\left(\rho_{2}\!-\!1\!-\!\dfrac{\lambda_1}{\mu}\right)L'_{T}(\lambda_1).
 \end{align}}

 {\textbf{Approximation 3:}
We assume that the queue is in the stationary state. In other words, first, we assume that the average residual service time of source $2$ packet that is  under service at the arrival instant of packet ${1,i}$ is equal to the average residual service time of a stationary M/G/1 queue that has only source 2 packet arrivals. Thus, we assume  that ${\mathbb{E}\left[R^{\mathrm{L}}_{2,i}|X_{1,i}=x,T_{1,i-1}=t,E_{1,i}^{\mathrm{L}}\right]=\dfrac{\lambda_2\mathbb{E}[S^2]}{2}}$ \cite[Eq. (3.52)]{bertsekas1992data}.   Second, we assume that the average number of source 2 packets that must be served before packet $1,i$ is equal to the average number of packets in a stationary M/G/1 queue with only source 2 packet arrivals. Thus, we assume $\mathbb{E}\left[M^{\mathrm{L}}_{2,i}|X_{1,i}=x,T_{1,i-1}=t,E_{1,i}^{\mathrm{L}}\right]= \dfrac{\lambda^2_2\mathbb{E}[S^2]}{2(1-\rho_2)}$ \cite[Eq. (3.43)]{bertsekas1992data}. Thus, the third conditional expectation in \eqref{mg1ma0in} is approximated as follows:
 \begin{align}\nonumber
 &\mathbb{E}[(S_{1,i}^{\mathrm{L}}\!+\!R^{\mathrm{L}}_{2,i})X_{1,i}|E_{1,i}^{\mathrm{L}}]\!\approx\!\dfrac{\lambda^2_2\mathbb{E}[S^2]}{2\mu(1-\rho_2)}\displaystyle\int_{0}^{\infty}\int_{0}^{\infty}xf_{X_{1,i},T_{1,i-1}|E_{1,i}^{\mathrm{L}}}(x,t)\mathrm{d}x\mathrm{d}t\\&\label{mg1mai0n}+\displaystyle\dfrac{\lambda_2\mathbb{E}[S^2]}{2}\int_{0}^{\infty}\int_{0}^{\infty}xf_{X_{1,i},T_{1,i-1}|E_{1,i}^{\mathrm{L}}}(x,t)\mathrm{d}x\mathrm{d}t\\&\nonumber\stackrel{}{=}
 \dfrac{\lambda_2\mathbb{E}[S^2]}{2(1-\rho_2)P(E^{\mathrm{L}}_{1,i})}\left( \dfrac{L_{T}(\lambda_1)}{\lambda_1}-L'_{T}(\lambda_1)\right).
 \end{align}
Using \eqref{mg1mai0n} and following the steps used to derive \eqref{mn01nm0}, an approximation for $ \mathbb{E}[X_{1,i}W_{1,i}] $ under Approximation 3 is given  as
\begin{align}\label{mn01nm03}
\mathbb{E}[X_{1,i}W_{1,i}] &\approx\dfrac{1}{\lambda_1}\bigg(\mathbb{E}[W]+\dfrac{1}{\mu}+\dfrac{2(\rho_{2}-1)}{\lambda_1}+ \left(\dfrac{\lambda_2\mathbb{E}[S^2]}{2(1-\rho_2)}+\dfrac{2(1-\rho_{2})}{\lambda_1}\right)L_{T}(\lambda_1)+\\&\nonumber\left(2\rho_{2}-1-\dfrac{\lambda_1\lambda_2\mathbb{E}[S^2]}{2(1-\rho_2)}\right)L'_{T}(\lambda_1)-\lambda_1\rho_2L''_{T}(\lambda_1)\bigg).
\end{align}
By substituting \eqref{mn01nm03} and \eqref{mbn0} into \eqref{oointr00103}, an approximation for the average AoI of source $1$ in a multi-source M/G/1 queueing model is given as
\begin{align}\label{main5433}
\Delta^{\text{app}_3}_{1}&\approx\mathbb{E}[W]+\dfrac{2}{\mu}+\dfrac{2\rho_{2}-1}{\lambda_1}+ \left(\dfrac{\lambda_2\mathbb{E}[S^2]}{2(1-\rho_2)}+\dfrac{2(1-\rho_{2})}{\lambda_1}\right)L_{T}(\lambda_1)+\\&\nonumber\left(2\rho_{2}-1-\dfrac{\lambda_1\lambda_2\mathbb{E}[S^2]}{2(1-\rho_2)}\right)L'_{T}(\lambda_1)-\lambda_1\rho_2L''_{T}(\lambda_1).
\end{align}}

{\subsection{Single-Source M/G/1 Queueing Model}
For $\lambda_2\rightarrow 0$, we have a single-source M/G/1 queueing model. In this case, it can be shown that  \eqref{main542} and \eqref{main5433} provide the following expression for the average AoI:
	\begin{align}\label{main54002}
	\Delta&=\mathbb{E}[W]+\dfrac{2}{\mu}+\dfrac{2L_{T}(\lambda_1)}{\lambda_1}-L'_{T}(\lambda_1)-\dfrac{1}{\lambda_1}.
	\end{align}
Using  \eqref{queue000}, \eqref{000bn0v001b0012}, and \eqref{rgn01}, the quantities $\mathbb{E}[W]$, $L_{T}(\lambda)$,  and $L'_{T}(\lambda)$ are calculated as
	$$\mathbb{E}[W]=\dfrac{\mathbb{E}[S^2]\lambda}{2(1-\rho)},{\quad}L_{T}(\lambda)=1-\rho,{\quad}L'_{T}(\lambda)=\dfrac{(1-\rho)(L_{S}(\lambda)-1)}{\lambda L_{S}(\lambda)}.$$
	By substituting $\mathbb{E}[W]$, $L_{T}(\lambda)$, and $L'_{T}(\lambda)$  in  \eqref{main54002},  we have
	\begin{align}\label{01mn}
	\Delta= \dfrac{1}{\mu}+\dfrac{\lambda \mathbb{E}[S^2]}{2(1-\rho)}+\dfrac{1-\rho}{\lambda L_S(\lambda)},
	\end{align}
	which is an \emph{exact} expression for the average AoI of the single-source M/G/1 queueing case derived in \cite[Eq.~(22)]{8006592}.}

\section{Validation and Simulation Results}\label{Validation and results}
In this section, we first evaluate the average AoI  in a multi-source M/M/1 queueing model and compare our exact expression in \eqref{mm10a0} with the results in existing works \cite{6284003} and \cite{ourITW01}. Then, we evaluate the accuracy of the proposed three approximate expressions for the M/G/1 queueing model in \eqref{main542}, \eqref{main543}, and \eqref{main5433} under various service time distributions.

\subsection{Multi-Source M/M/1 Queueing Model}
Fig. \ref{multi100} depicts the average AoI of source $1$ ($\Delta_1$) as a function of $\lambda_1$ with $\lambda_2=0.6$ and ${\mu}=1$. As it can be seen, the simulation result and our proposed solution overlap perfectly. {We used \textit{``integral2''} command in MATLAB software to calculate the double integral in \eqref{01mnk}.} Due to the calculation errors in \cite{6284003} and \cite{ourITW01}, both curves have a gap to the correct average AoI value.

\begin{figure}[t!]
\centering
\includegraphics[width=0.45\linewidth,trim = 5mm 2mm 5mm 10mm,clip]{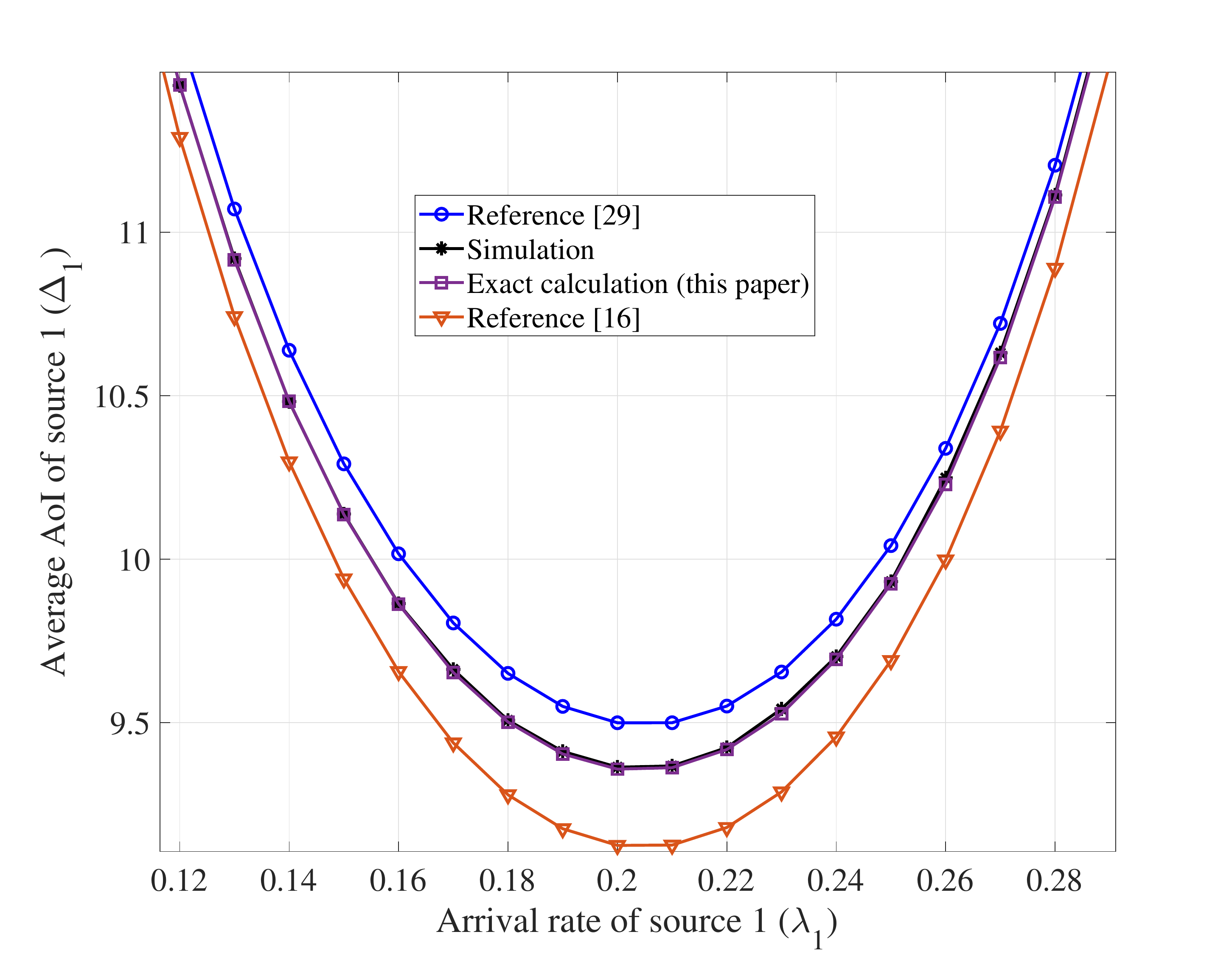}
\vspace{-3mm}
\caption{The average AoI of  source 1 as a function of  $\lambda_1$ with $\lambda_2=0.6$ and ${\mu}=1$.  }
\vspace{-9mm}
\label{multi100}
\end{figure}

The effect of  $\lambda_2$ on the average AoI  of source 1 is shown in Fig. \ref{multi10}. When $\lambda_2$ increases, the increased overall load in the system results in longer waiting time for packets of source 1 (and source 2), which increases $\Delta_1$. Note, however, that when $\lambda_2$ increases, the optimal value of $\lambda_1$ that minimizes $\Delta_1$ decreases. The figures illustrate that generating the status update packets too frequently or too rarely does not minimize the average AoI. 
Moreover, Fig. \ref{multi10} depicts the gap between the exact and approximate average AoI expressions. As it can be seen, the proposed approximations are relatively close to the exact one in the M/M/1 queueing model.

\begin{figure}[t!]
\centering
{\includegraphics[width=0.47\linewidth,trim = 5mm 0mm 5mm 5mm,clip]{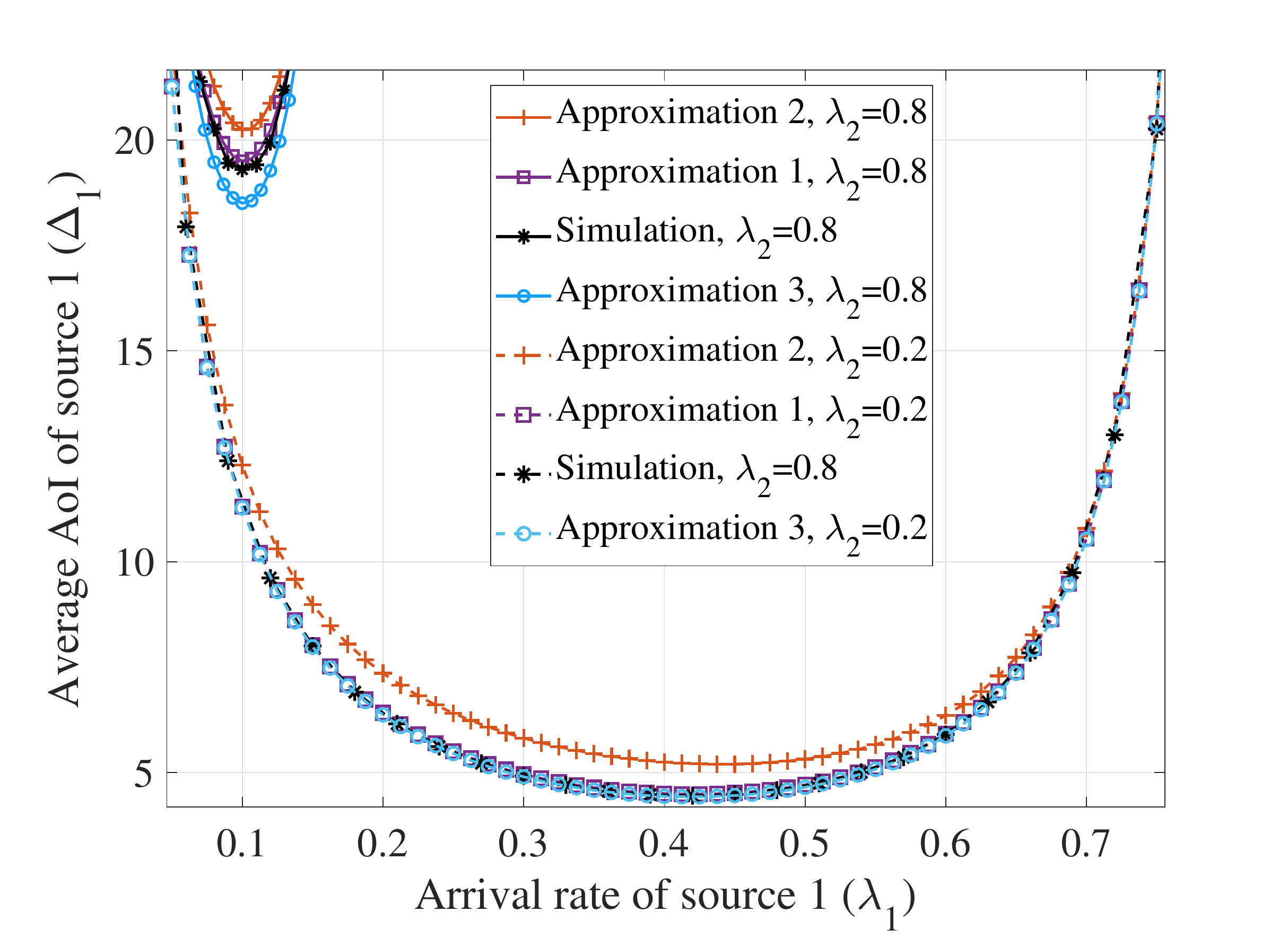}}\vspace{-2mm}
\caption{The average AoI of source 1 as a function of $\lambda_1$ for different values of $\lambda_2$ with $\mu=1$.}\vspace{-8mm}  			
\label{multi10}
\end{figure}

Fig. \ref{multi12} depicts the average delay of source 1 as a function of $\lambda_1$ for different values of $\lambda_2$ with $\mu=1$. The average delay is defined as the summation of the average waiting time and average service time i.e., $\mathbb{E}[W]+{1}/{\mu}$. As the number of arrivals of source 2 packets increases, the queue becomes more congested and the average delay of source 1 increases. By comparing Figs. \ref{multi10} and \ref{multi12} one can see that the delay does not fully capture the information freshness, i.e., minimizing the average system delay does not necessarily lead to a good performance in terms of AoI and, reciprocally,  minimizing the average AoI does not minimize the average system delay.

\begin{figure}
\centering
\includegraphics[width=0.45\linewidth,trim = 5mm 2mm 5mm 5mm,clip]{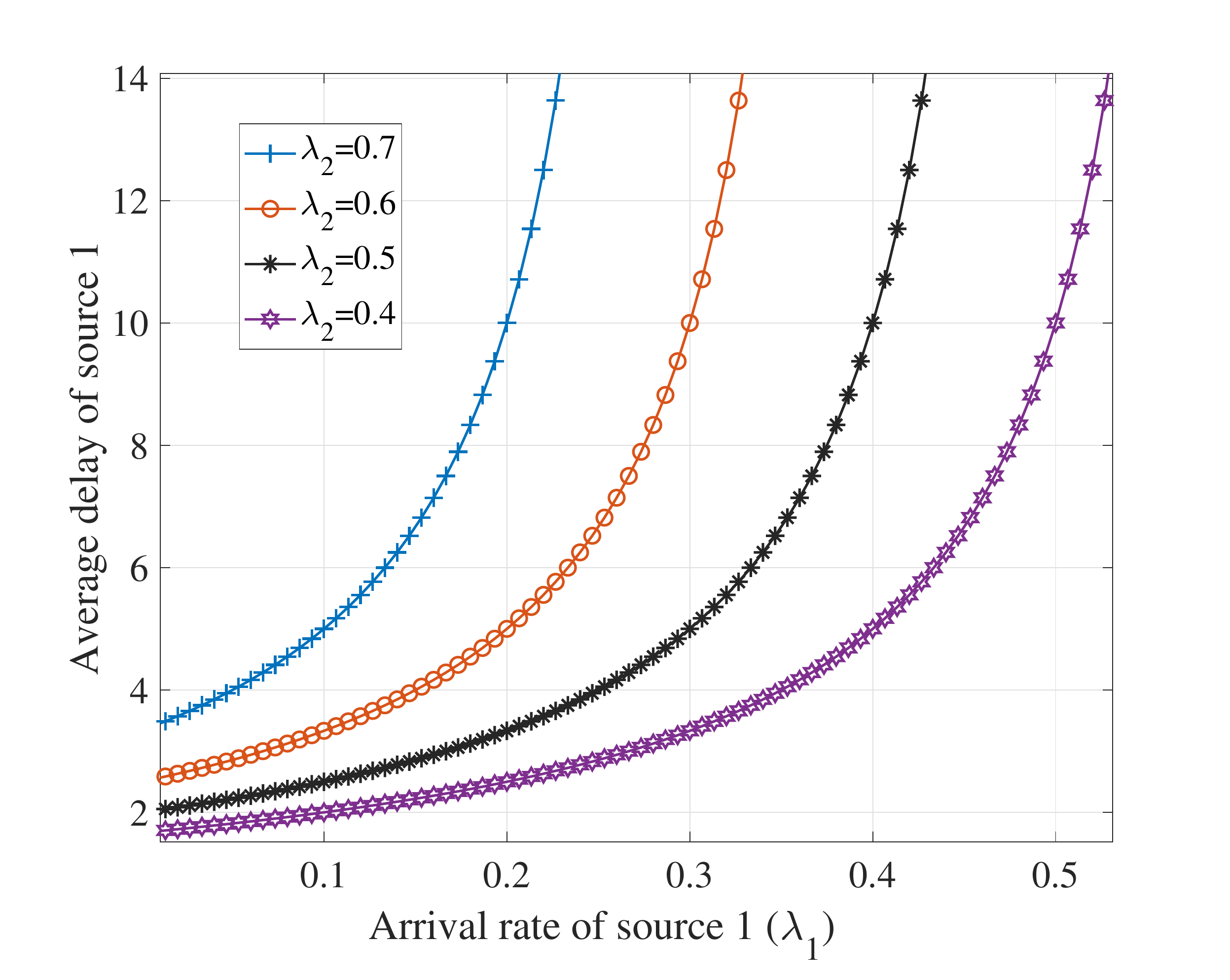}\vspace{-1mm}			 	
\caption{The average delay of  source 1 as a function of  $\lambda_1$ for different values of $\lambda_2$ with $\mu=1$.}
\vspace{-9mm}
\label{multi12}
\end{figure}

\subsection{Multi-Source M/G/1 Queueing Model}
In the section, we examine the accuracy of the proposed three approximations using the following service time distributions: i) Gamma distribution, ii) hyper-exponential distribution, iii) log-normal distribution, and iv)  Pareto distribution. In the following, we first define the distributions and then show the accuracy of the proposed approximations for each distribution.
{\textbf{Definition 1} (Gamma distribution).
The PDF of  a random variable $S$ following a  gamma distribution is defined as
$
f_S(s)=Gamma(s;\kappa,\beta)=\dfrac{\beta^{\kappa} s^{\kappa-1}\exp(-\beta s)}{\Gamma(\kappa)},~ \text{for}~s>0,$ and parameters $ \kappa>0$ and $\beta>0,
$ where $\Gamma(\kappa)$ is the gamma function at $\kappa$. The mean and variance of this random variable is $\mathbb{E}[S]={\kappa}/{\beta}$ and $\text{Var}[S]=\dfrac{\kappa}{\beta^2}$, respectively.}

{\textbf{Definition 2} (Hyper-exponential distribution).
The PDF of a random variable $S$ following a hyper-exponential distribution is defined as
$
f_S(s)=\sum_{k=1}^{N}f_{Y_k}(s)p_k,
$
where $Y_k$ is an exponentially distributed random variable with parameter $\gamma_k$, and $p_k$ is the weight factor of random variable $Y_k$ such that $\sum_{k=1}^{N}p_k=1$.  The mean and variance of this random variable are
$
\mathbb{E}[S]=\displaystyle\sum_{k=1}^{N}\dfrac{p_k}{\gamma_k}$ 
and 
$
\text{Var}[S]=\sum_{k=1}^{N}\dfrac{2p_k}{\gamma^2_k}-\left(\sum_{k=1}^{N}\dfrac{p_k}{\gamma_k}\right)^2,
$ respectively.}

{\textbf{Definition 3} (Log-normal distribution).
The PDF of a random variable $S$ following a log-normal distribution is defined as
$
f_S(s)=\dfrac{1}{s\sigma\sqrt{2\pi}}\exp\left(-\dfrac{(\ln{(s)}-\nu)^2}{2\sigma^2}\right),~ \text{for}~s>0$ and parameters $ \sigma>0$ and $\nu\in(-\infty,+\infty).
$
 The mean and variance of this random variable are 
$
\mathbb{E}[S]=\exp\left(\nu+\dfrac{\sigma^2}{2}\right)$ and $
\text{Var}[S]=\exp(2\nu+\sigma^2)\left(\exp(\sigma^2)-1\right),
$  respectively.
}

{\textbf{Definition 4} (Pareto distribution).
The PDF of a random variable $S$ following a Pareto distribution is defined as
$
f_S(s)=\dfrac{\alpha {\omega}^\alpha}{s^{\alpha+1}},\,\,\, \text{for}\,\,\,s\in[\omega,\infty]$ and parameters $ \omega>0$ and $\alpha>0.$
 The mean and variance of this random variable are 
\begin{align}\nonumber
&\mathbb{E}[S]=\begin{cases}
\infty&\alpha\le 1\\
\dfrac{\alpha \omega}{\alpha-1}&\alpha>1,
\end{cases}~~
\text{Var}[S]=\begin{cases}
\infty&\alpha\le 2\\
\dfrac{\alpha \omega^2}{(\alpha-1)^2(\alpha-2)}&\alpha>2.
\end{cases}
\end{align}
}


\begin{figure}
	\centering
	\subfloat[]{\includegraphics[width=0.45\linewidth,trim = 8mm 0mm 10mm 8mm,clip]{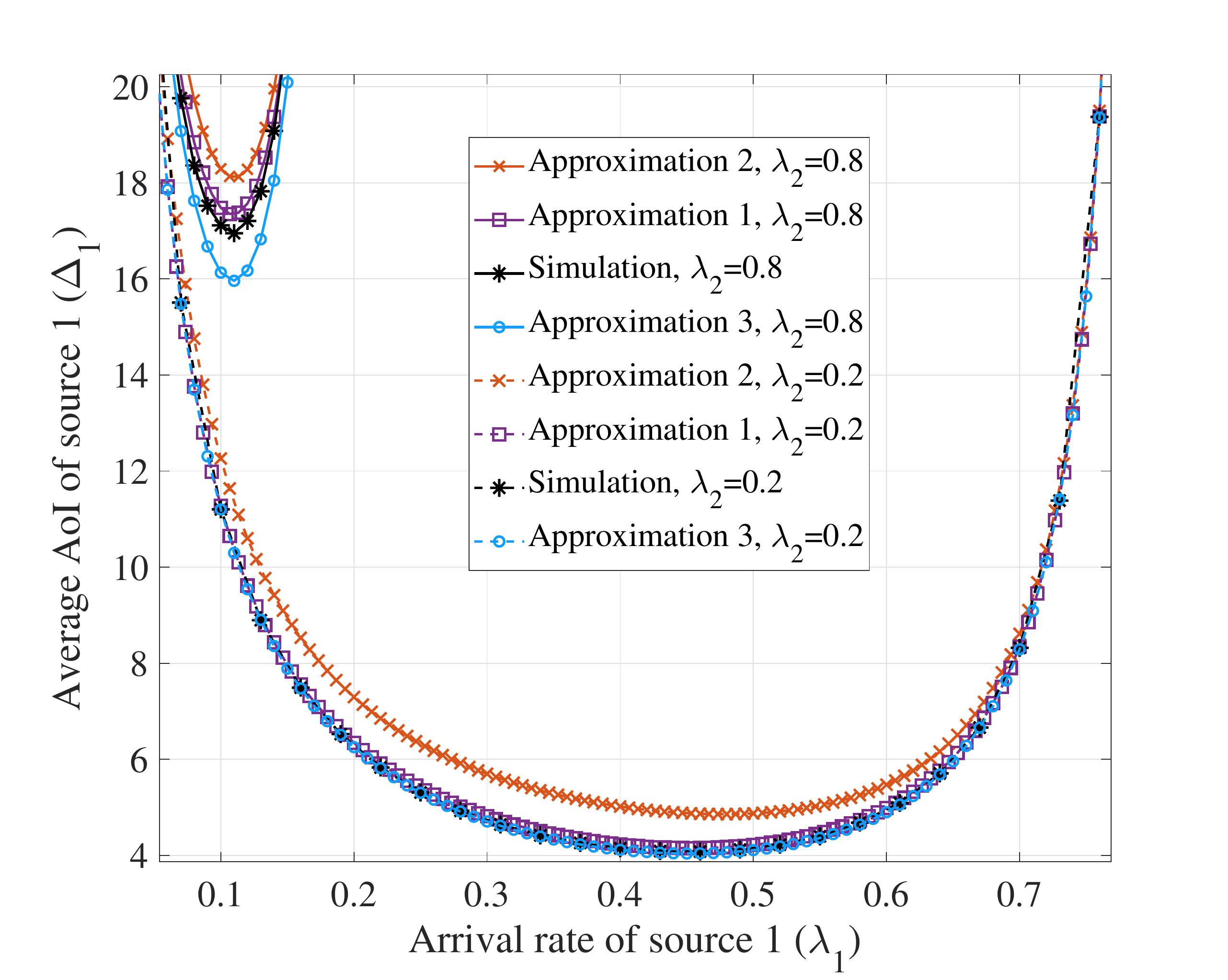}}\\\vspace{-4mm}
	\subfloat[]{\includegraphics[width=0.45\linewidth,trim = 6mm 0mm 10mm 8mm,clip]{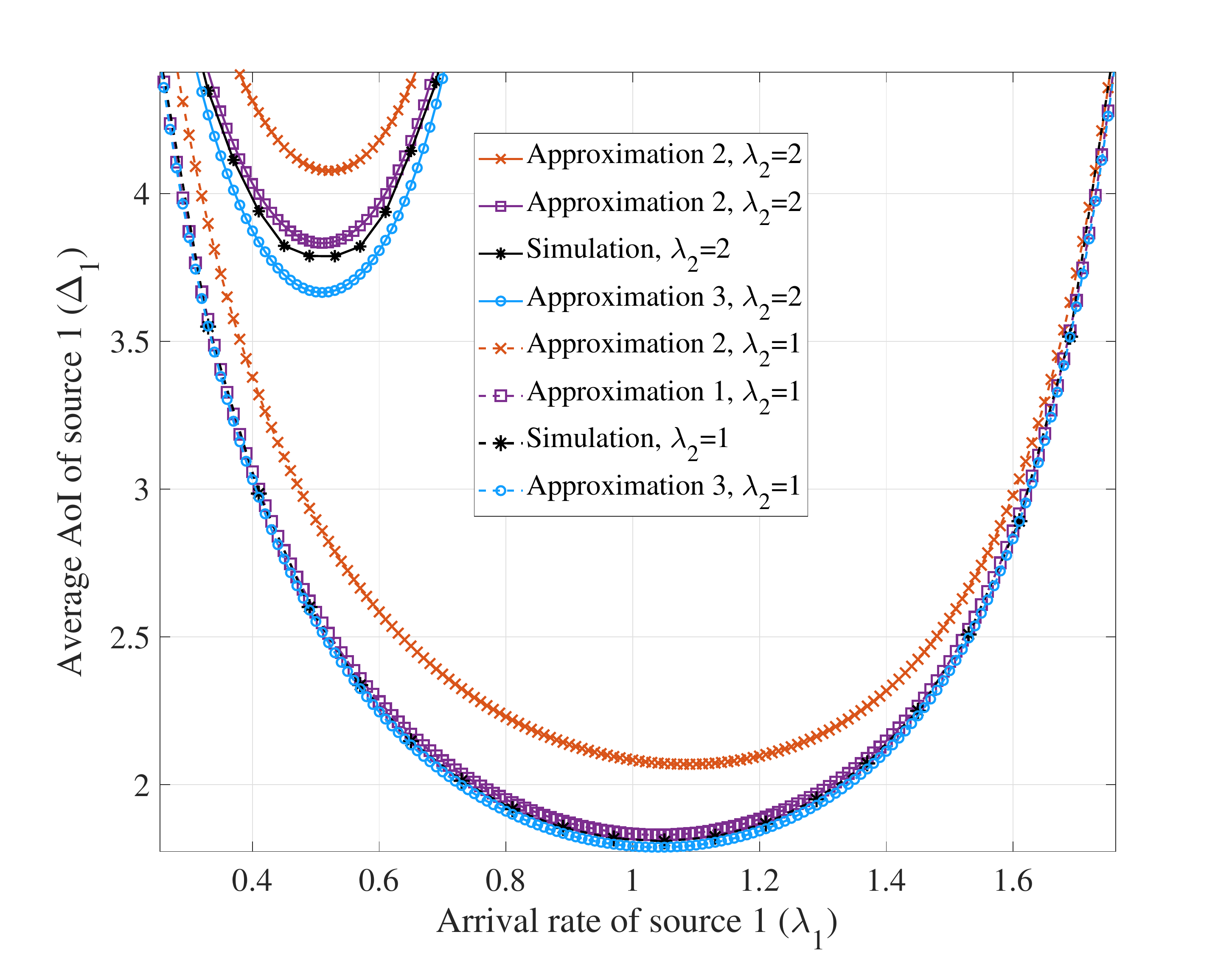}}\vspace{-4mm}
	\caption{The average AoI of  source 1 as a function of  $\lambda_1$ for different values  of  $\lambda_2$   with  the service time following a gamma distribution with parameters (a) $\kappa=2,~ \beta=2,$ and $\mu=1$,  and (b) $\kappa=1,~ \beta=3,$ and $ \mu=3$.}	
	\vspace{-9mm}		
	\label{Gamma12}
\end{figure}

\begin{figure}
	\centering
	\subfloat[]{\includegraphics[width=0.43\linewidth,trim = 10mm 0mm 15mm 10mm,clip]{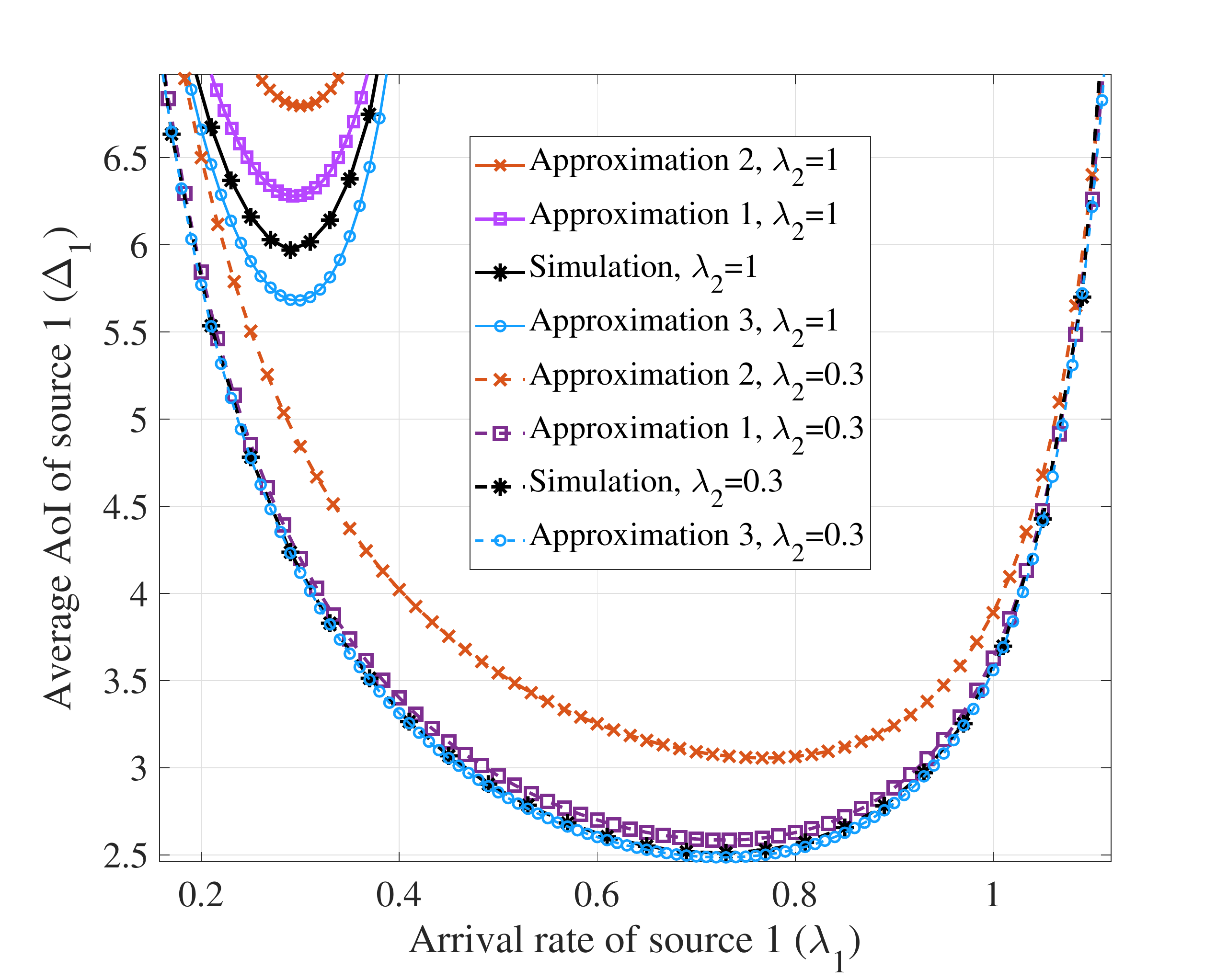}}\\\vspace{-4mm}
	\subfloat[]{\includegraphics[width=0.45\linewidth,trim = 8mm 0mm 15mm 10mm,clip]{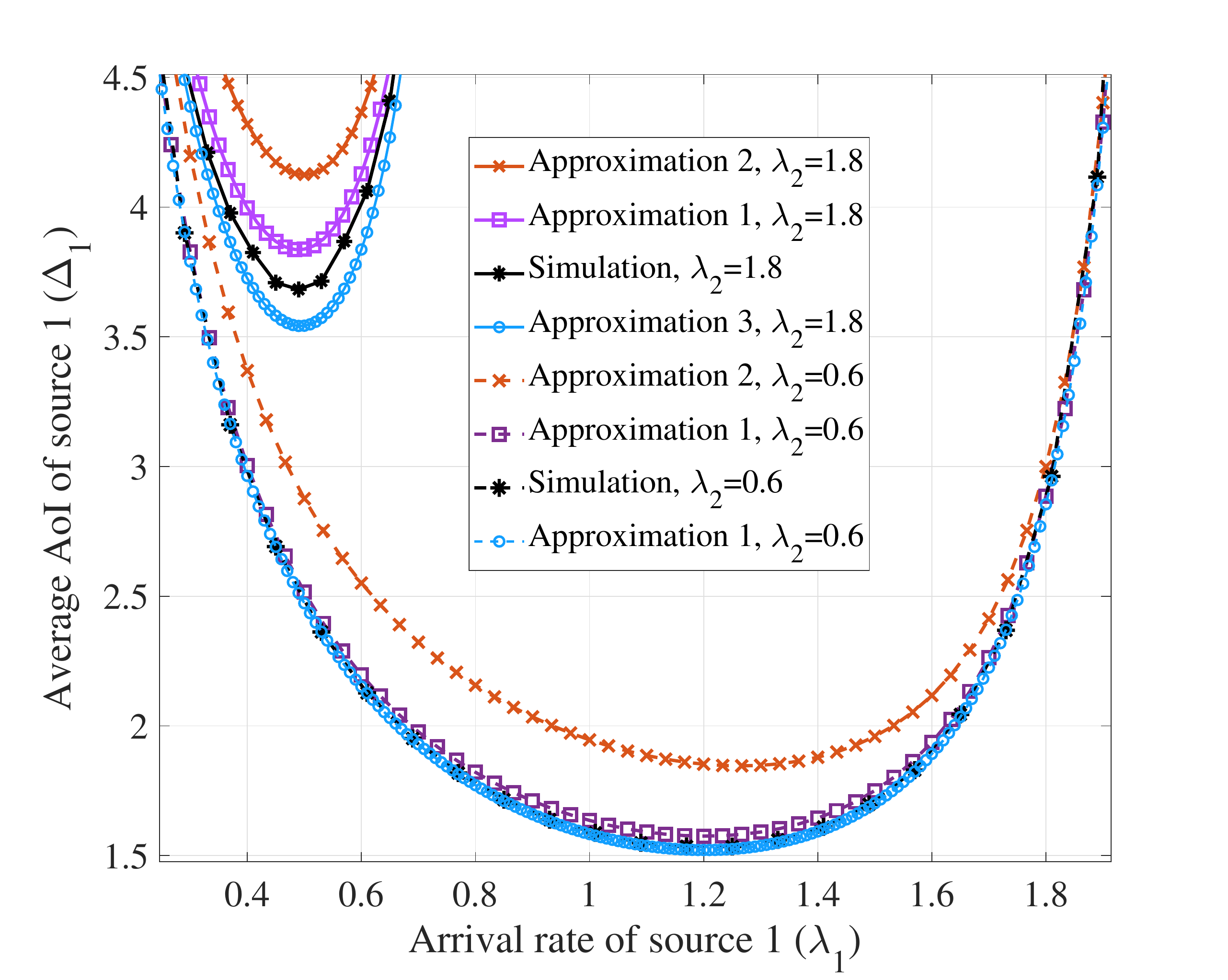}}\vspace{-4mm}
	\caption{The average AoI of  source 1 as a function of  $\lambda_1$ for different values  of  $\lambda_2$   with the  service time following a Pareto distribution  with parameters (a) $\omega=0.5,~ \alpha=4,$ and $ \mu=1.5$,  and (b) $\omega=0.25,~ \alpha=3,$ and $ \mu=8/3$.}
		\vspace{-10mm}					
	\label{Pareto12}
\end{figure}

\begin{figure}
	\centering
	\subfloat[]{\includegraphics[width=0.45\linewidth,trim = 10mm 0mm 15mm 10mm,clip]{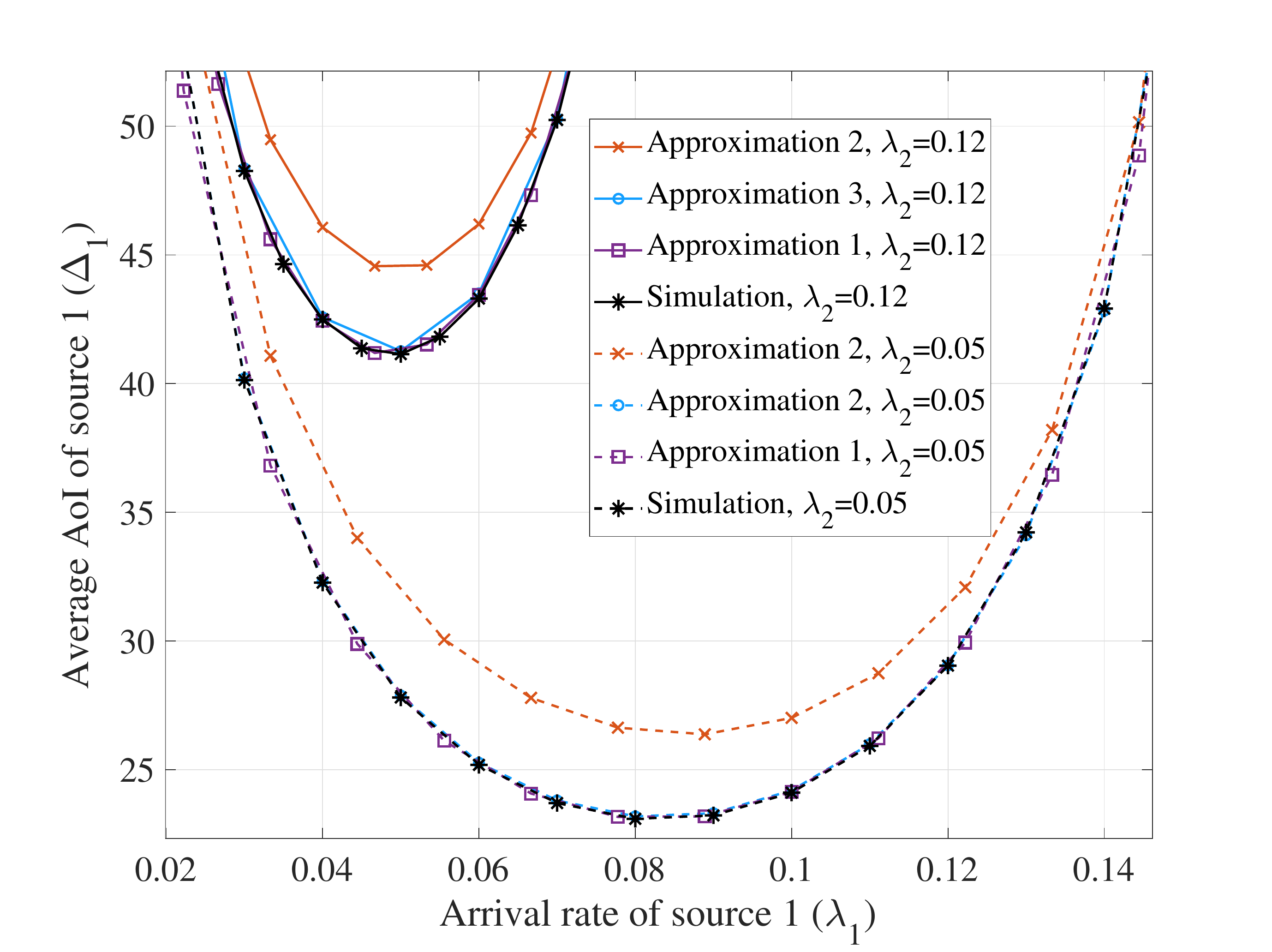}}\\\vspace{-4mm}
	\subfloat[]{\includegraphics[width=0.45\linewidth,trim = 8mm 0mm 15mm 10mm,clip]{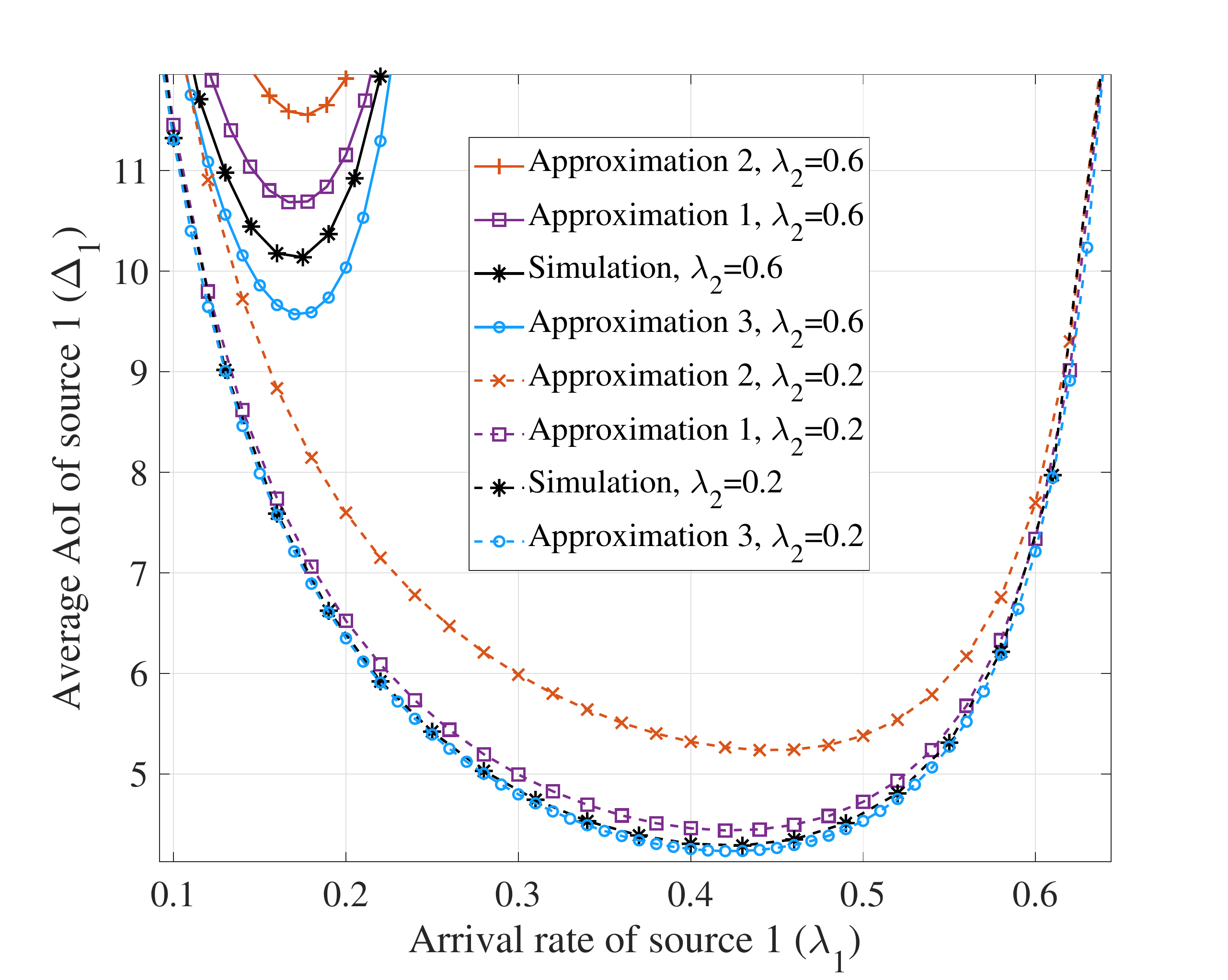}}\vspace{-4mm}
	\caption{The average AoI of  source 1 as a function of  $\lambda_1$ for different values  of  $\lambda_2$   with  the service time following a log-normal distribution with parameters (a) $\nu=1,~ \sigma=1,$ and $\mu=0.2231$,  and (b) $\nu=0.1,~ \sigma=0.2,$ and $ \mu=0.8869$.}	
		\vspace{-12mm}				
	\label{log12}
\end{figure}

\begin{figure}
	\centering
	\subfloat[]{\includegraphics[width=0.45\linewidth,trim = 10mm 0mm 15mm 10mm,clip]{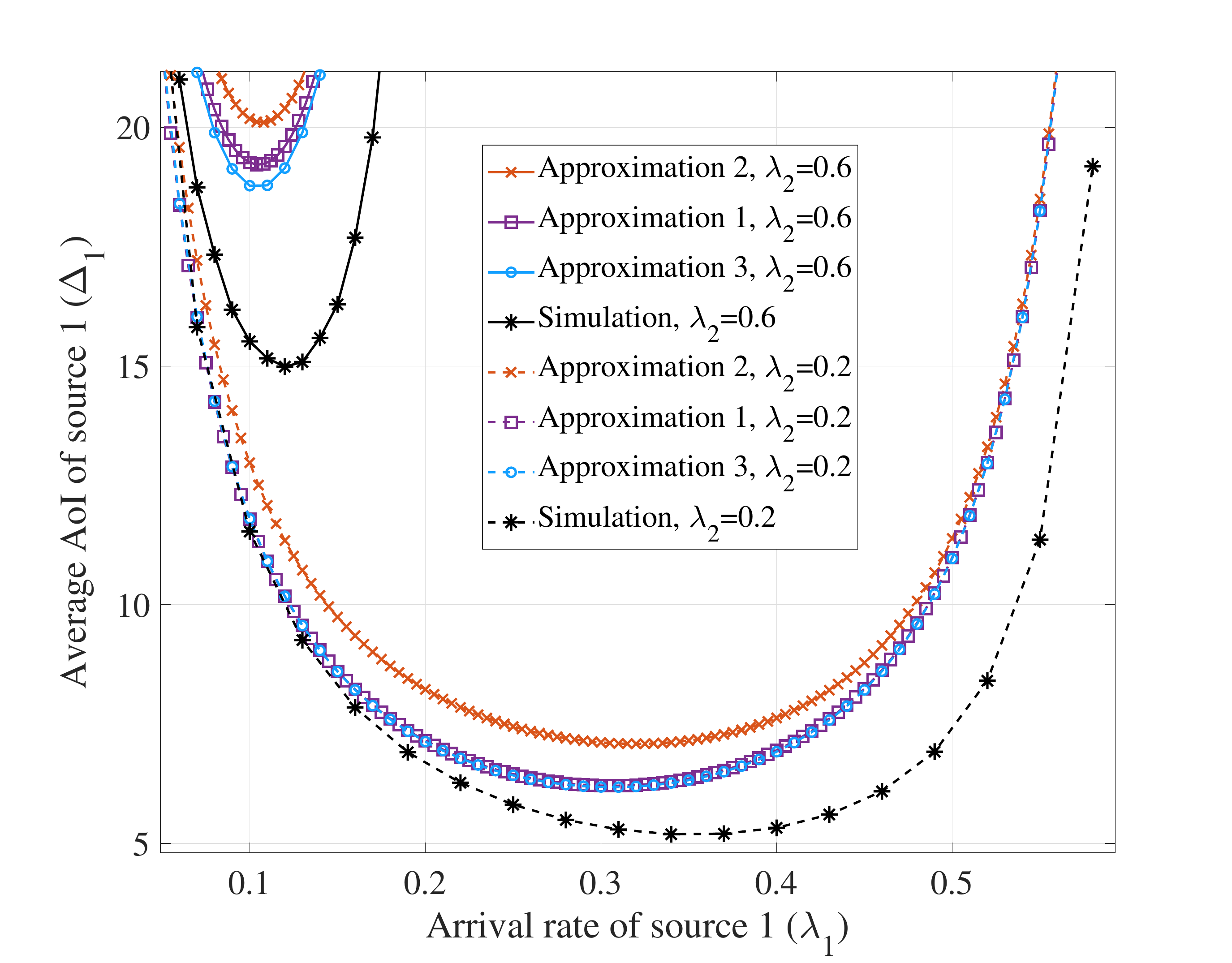}}\\
	\subfloat[]{\includegraphics[width=0.45\linewidth,trim = 8mm 0mm 15mm 10mm,clip]{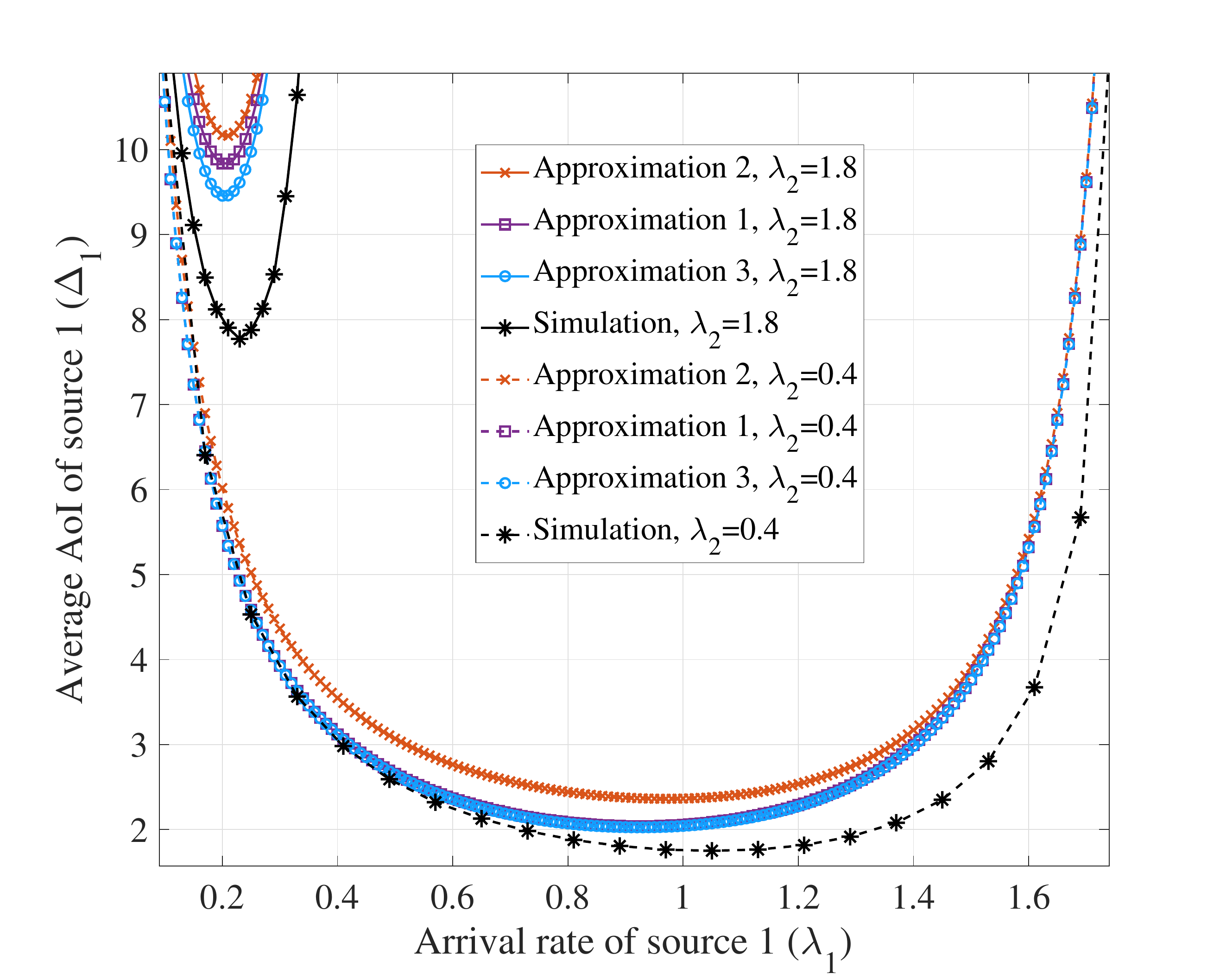}}\vspace{-3mm}
	\caption{The average AoI of  source 1 as a function of  $\lambda_1$ for different values  of  $\lambda_2$   with  the service time following a hyper-exponential distribution with parameters (a) $N=3,~\gamma_1=0.5,~\gamma_2=1,~\gamma_3=1.5, p_k=1/N,~\forall k\in\{1,2,3\},$ and $ \mu=0.8182$,  and (b) $N=3,~\gamma_1=1.5,~\gamma_2=2.5,~\gamma_3=3.5, p_k=1/N,~\forall k\in\{1,2,3\},$ and $ \mu=2.2183$.}
	\vspace{-9mm}			
	\label{Hyper12}
\end{figure}

{Figs. \ref{Gamma12}, \ref{Pareto12}, \ref{log12}, and \ref{Hyper12} depict the average AoI of source 1 as a function of  $\lambda_1$ for different service time distributions  under both heavy (a larger value of $\lambda_2$) and light (a smaller value of $\lambda_2$) traffic conditions of source 2. 
	Fig. \ref{Gamma12} illustrates the average AoI of  source 1  for different values  of  $\lambda_2$   with the service time following a gamma distribution with parameters   $\kappa=2, ~\beta=2,$ and $\mu=1$ in Fig. \ref{Gamma12}(a)  and  $\kappa=1,~ \beta=3,$ and  $\mu=3$ in Fig. \ref{Gamma12}(b). 
	Fig. \ref{Pareto12} illustrates the average AoI of  source 1  for different values  of  $\lambda_2$   with the service time following a Pareto distribution with parameters  $\omega=0.5,~ \alpha=4,$ and $\mu=1.5$ in Fig. \ref{Pareto12}(a)  and  $\omega=0.25,~ \alpha=3,$ and $\mu=8/3$ in Fig. \ref{Pareto12}(b). 
	Fig. \ref{log12} illustrates the average AoI of  source 1  for different values  of  $\lambda_2$   with  the  service time following a log-normal distribution with parameters  $\nu=1,~ \sigma=1,$ and  $\mu=0.2231$ in Fig. \ref{log12}(a)  and  $\nu=0.1,~ \sigma=0.2,$ and  $\mu=0.8869$ in Fig. \ref{log12}(b). 
	Fig. \ref{Hyper12} illustrates the average AoI of  source 1  for different values  of  $\lambda_2$   with  the service time following a hyper-exponential distribution with parameters  $N=3,~\gamma_1=0.5,~\gamma_2=1,~\gamma_3=1.5, p_k=1/N,~\forall k\in\{1,2,3\},$ and $\mu=0.8182$ in Fig. \ref{Hyper12}(a)  and  $N=3,~\gamma_1=1.5,~\gamma_2=2.5,~\gamma_3=3.5, p_k=1/N,~\forall k\in\{1,2,3\},$ and $ \mu=2.2183$ in Fig. \ref{Hyper12}(b).} 
{ As it can be seen, Approximation 1 and Approximation 3  are relatively tight  for  both the heavy and light traffic conditions under the gamma, Pareto, and log-normal distributions. 
	By comparing the curves of Approximation 1 and Approximation 2, we can see the effect of approximating the residual service time of source $2$ packet that is under service at the arrival instant of packet ${1,i}$ by the average service time of one packet in the system as compared to completely ignoring it. Finally, as expected, the average AoI provided by Approximation 2 is always higher than that of Approximation 1.}

\section{Conclusions}\label{conclusion}
We considered a single-server multi-source FCFS  queueing model with Poisson arrivals and analyzed the average AoI of each source. We derived 1) an exact expression for the average AoI for a multi-source M/M/1 queueing model and 2) three approximate expressions for the average AoI for a multi-source M/G/1 queueing model. The simulation results showed that the approximate expressions for the average AoI are relatively accurate for different service time distributions.
In addition, the results pointed out the significance of the AoI as a metric in time-sensitive control applications: minimizing merely the average delay does not minimize the AoI.

\vspace{-5mm}
\section*{Acknowledgements}
{The authors would like to thank Prof. Roy Yates and the anonymous reviewers for pointing out errors in our initial manuscript and providing  invaluable suggestions for improving the paper.} 

 This research has been financially supported by the Infotech Oulu, the Academy of Finland (grant 323698), and Academy of Finland 6Genesis Flagship (grant 318927). M. Codreanu would like to acknowledge the support of the European Union's Horizon 2020 research and innovation programme under the Marie Sk\l{}odowska-Curie Grant Agreement No. 793402 (COMPRESS NETS). M. Moltafet would like to acknowledge the support of Finnish Foundation for Technology Promotion and HPY Research Foundation.

%

\appendices
{\section{Proof of Lemma \ref{lemmii1}, \ref{gh0fd0}, and \ref{nvbfh}}\label{Lemapp}
\subsection{Proof of Lemma \ref{lemmii1}} \label{Lem1app}
Using the facts that $T_{1,i-1}$ and $X_{1,i}$  are independent and the  PDF of $X_{1,i}$ is ${f_{X_{1,i}}(x)=\lambda_1 e^{-\lambda_1 x}}$,  $P(E^{\mathrm{B}}_{1,i})$ can be written as
	\begin{align}\label{nmhg}
	P(E^{\mathrm{B}}_{1,i})&=\int_{0}^{\infty}P(T_{1,i-1}\ge X_{1,i}|T_{1,i-1}=t)f_{T_{1,i-1}}(t)\mathrm{d}t\\&\nonumber
	=1-\int_{0}^{\infty} e^{-\lambda_1t}f_{T_{1,i-1}}(t)\mathrm{d}t
	\stackrel{(a)}{=}1-L_{T}(\lambda_1),
	\end{align}
	where	equality (a) follows because the system times of different packets are stochastically identical, i.e., $T_{1,i}=^{\mathrm{st}}T_{2,i}=^{\mathrm{st}}T$, $\forall{i}$ \cite{8187436,6284003}; 
	and $L_{T}(\lambda_1)$  denotes the Laplace transform of the PDF of the system time $T$ at $\lambda_1$. Because  $E^{\mathrm{L}}_{1,i}$ is the complementary  event of $E^{\mathrm{B}}_{1,i}$, we have
	\begin{equation}\label{nmhg1}
	P(E^{\mathrm{L}}_{1,i})=1-P(E^{\mathrm{B}}_{1,i})=L_{T}(\lambda_1).
	\end{equation}}

The relation between the Laplace transforms of the PDFs of the system time $T$ and service time $S$ is given as  \cite[Sect.~5.1.2]{daigle2005basic}
\begin{align}\label{0bn0v001b0012}
L_{T}(a)=\dfrac{\big(1-\rho\big)aL_{S}(a)}{a-\lambda\big(1-L_{S}(a)\big)}.
\end{align}
Finally, substituting \eqref{0bn0v001b0012} in \eqref{nmhg} and \eqref{nmhg1} results in the expressions \eqref{proffe01} and \eqref{proffe02}, respectively.

\subsection{Proof of Lemma \ref{gh0fd0}} \label{Lem2app}
To prove Lemma \ref{gh0fd0}, we use the fact that for random variables $ Y_1 $ and $ Y_2 $ and a certain event $ A $, the conditional PDF $f_{Y_1,Y_2|A}(y_1,y_2)$ is given by	\cite[Sect.~4.4]{papbook01}
\begin{align}\label{nvhg01002}
f_{Y_1,Y_2|A}(y_1,y_2)  =
\begin{cases}\dfrac{f_{Y_1,Y_2}(y_1,y_2)}{P(A)}&(y_1,y_2)\in A\\0&\text{otherwise}.\end{cases}
\end{align}	
In Lemma \ref{gh0fd0}, $ Y_1 $ and $ Y_2 $ are $X_{1,i}$ and $T_{1,i-1}$, respectively, which are two independent random variables, and event $A$ is $E^{\mathrm{B}}_{1,i}$.

\subsection{Proof of Lemma \ref{nvbfh}} \label{Lem3app}
According to the feature of the Laplace transform, for any function $f(y), y\ge 0$, we have \cite[Sect.~13.5] {rade2013mathematics}:
\begin{align}\label{nbhg01}
L_{\int_{0}^{y}f(b)\mathrm{d}b}(a)=\dfrac{L_{f(y)}(a)}{a}.
\end{align}
Therefore, using \eqref{dr001} and \eqref{nbhg01}, we have
\begin{align}
L_{x^2F_{T_1}(x)}(a)\bigg|_{a=\lambda_1}&= L_{x^2\int_{0}^{x}f_{T_1}(b)\mathrm{d}b}(a)\bigg|_{a=\lambda_1}=\dfrac{\mathrm{d}^2\bigg(\dfrac{L_{T}(a)}{a}\bigg)}{\mathrm{d}a^2}\bigg|_{a=\lambda_1}\\&\nonumber= \dfrac{a L''_{T}(a)-2L'_{T}(a)}{a^2}+\dfrac{2L_{T}(a)}{a^3}\bigg|_{a=\lambda_1}\stackrel{}{=} \dfrac{\lambda_1 L''_{T}(\lambda_1)-2L'_{T}(\lambda_1)}{\lambda_1^2}+\dfrac{2L_{T}(\lambda_1)}{\lambda_1^3}.
\end{align}

\bibliographystyle{IEEEtran}
\begin{spacing}{1.7}
\bibliography{RBibliography}

\begin{thebibliography}{10}
\providecommand{\url}[1]{#1}
\csname url@samestyle\endcsname
\providecommand{\newblock}{\relax}
\providecommand{\bibinfo}[2]{#2}
\providecommand{\BIBentrySTDinterwordspacing}{\spaceskip=0pt\relax}
\providecommand{\BIBentryALTinterwordstretchfactor}{4}
\providecommand{\BIBentryALTinterwordspacing}{\spaceskip=\fontdimen2\font plus
\BIBentryALTinterwordstretchfactor\fontdimen3\font minus
  \fontdimen4\font\relax}
\providecommand{\BIBforeignlanguage}[2]{{%
\expandafter\ifx\csname l@#1\endcsname\relax
\typeout{** WARNING: IEEEtran.bst: No hyphenation pattern has been}%
\typeout{** loaded for the language `#1'. Using the pattern for}%
\typeout{** the default language instead.}%
\else
\language=\csname l@#1\endcsname
\fi
#2}}
\providecommand{\BIBdecl}{\relax}
\BIBdecl

\bibitem{corke2010environmental}
P.~Corke, T.~Wark, R.~Jurdak, W.~Hu, P.~Valencia, and D.~Moore, ``Environmental
  wireless sensor networks,'' \emph{Proc. IEEE}, vol.~98, no.~11, pp.
  1903--1917, Nov. 2010.

\bibitem{5307471}
P.~Papadimitratos, A.~D.~L. Fortelle, K.~Evenssen, R.~Brignolo, and S.~Cosenza,
  ``Vehicular communication systems: Enabling technologies, applications, and
  future outlook on intelligent transportation,'' \emph{{IEEE} Commun. Mag.},
  vol.~47, no.~11, pp. 84--95, Nov. 2009.

\bibitem{818826}
M.~Xiong and K.~Ramamritham, ``Deriving deadlines and periods for real-time
  update transactions,'' in \emph{Proc. IEEE Real. Time. Sys. Symp.}, Phoenix,
  AZ, USA, Dec. 1--3, 1999, pp. 32--43.

\bibitem{hu2002ensuring}
Y.-C. Hu and D.~B. Johnson, ``Ensuring cache freshness in on-demand ad hoc
  network routing protocols,'' in \emph{Proc. IEEE Princ. of Mobile. Comp.},
  Toulouse, France, Oct. 30--31, 2002, pp. 25--30.

\bibitem{8187436}
A.~{Kosta}, N.~{Pappas}, and V.~{Angelakis}, ``Age of information: A new
  concept, metric, and tool,'' \emph{Foun. and Trends in Net.}, vol.~12, no.~3,
  pp. 162--259, 2017.

\bibitem{8469047}
R.~D. {Yates} and S.~K. {Kaul}, ``The age of information: Real-time status
  updating by multiple sources,'' \emph{{IEEE} Trans. Inform. Theory}, vol.~65,
  no.~3, pp. 1807--1827, Mar. 2019.

\bibitem{6195689}
S.~Kaul, R.~Yates, and M.~Gruteser, ``Real-time status: How often should one
  update?'' in \emph{Proc. IEEE Int. Conf. on Computer. Commun. (INFOCOM)},
  Orlando, FL, USA, Mar. 25--30, 2012, pp. 2731--2735.

\bibitem{6310931}
S.~K. Kaul, R.~D. Yates, and M.~Gruteser, ``Status updates through queues,'' in
  \emph{Proc. Conf. Inform. Sciences Syst. (CISS)}, Princeton, NJ, USA, Mar.
  21--23, 2012, pp. 1--6.

\bibitem{5984917}
S.~Kaul, M.~Gruteser, V.~Rai, and J.~Kenney, ``Minimizing age of information in
  vehicular networks,'' in \emph{Proc. Commun. Society. Conf. on Sensor, Mesh
  and Ad Hoc Commun. and Net.}, Salt Lake City, UT, USA, Jun. 27--30, 2011, pp.
  350--358.

\bibitem{8006593}
A.~M. Bedewy, Y.~Sun, and N.~B. Shroff, ``Age-optimal information updates in
  multihop networks,'' in \emph{Proc. IEEE Int. Symp. Inform. Theory}, Aachen,
  Germany, Jun. 25--30, 2017, pp. 576--580.

\bibitem{6875100}
M.~Costa, M.~Codreanu, and A.~Ephremides, ``Age of information with packet
  management,'' in \emph{Proc. IEEE Int. Symp. Inform. Theory}, Honolulu, HI,
  USA, Jun. 20--23, 2014, pp. 1583--1587.

\bibitem{7415972}
------, ``On the age of information in status update systems with packet
  management,'' \emph{{IEEE} Trans. Inform. Theory}, vol.~62, no.~4, pp.
  1897--1910, Apr. 2016.

\bibitem{8006504}
E.~Najm, R.~Yates, and E.~Soljanin, ``Status updates through {M/G/1/1} queues
  with {HARQ},'' in \emph{Proc. IEEE Int. Symp. Inform. Theory}, Aachen,
  Germany, Jun. 25--30, 2017, pp. 131--135.

\bibitem{8006592}
Y.~Inoue, H.~Masuyama, T.~Takine, and T.~Tanaka, ``The stationary distribution
  of the age of information in {FCFS} single-server queues,'' in \emph{Proc.
  IEEE Int. Symp. Inform. Theory}, Aachen, Germany, Jun. 25--30, 2017, pp.
  571--575.

\bibitem{7541764}
E.~{Najm} and R.~{Nasser}, ``Age of information: The gamma awakening,'' in
  \emph{Proc. IEEE Int. Symp. Inform. Theory}, Barcelona, Spain, Jul. 10--16,
  2016, pp. 2574--2578.

\bibitem{6284003}
R.~D. Yates and S.~Kaul, ``Real-time status updating: Multiple sources,'' in
  \emph{Proc. IEEE Int. Symp. Inform. Theory}, Cambridge, MA, USA, Jul. 1--6,
  2012, pp. 2666--2670.

\bibitem{7282742}
L.~Huang and E.~Modiano, ``Optimizing age-of-information in a multi-class
  queueing system,'' in \emph{Proc. IEEE Int. Symp. Inform. Theory}, Hong Kong,
  China, Jun. 14--19, 2015, pp. 1681--1685.

\bibitem{8406928}
E.~Najm and E.~Telatar, ``Status updates in a multi-stream {M/G/1/1} preemptive
  queue,'' in \emph{Proc. IEEE Int. Conf. on Computer. Commun. (INFOCOM)},
  Honolulu, HI, USA, Apr. 15--19, 2018, pp. 124--129.

\bibitem{Kosta2018AgeOI}
A.~Kosta, N.~Pappas, A.~Ephremides, and V.~Angelakis, ``Age of information and
  throughput in a shared access network with heterogeneous traffic,'' in
  \emph{Proc. IEEE Global Telecommun. Conf.}, Abu Dhabi, United Arab Emirates,
  Dec. 9--13, 2018, pp. 1--6.

\bibitem{8006544}
R.~D. Yates and S.~K. Kaul, ``Status updates over unreliable multiaccess
  channels,'' in \emph{Proc. IEEE Int. Symp. Inform. Theory}, Aachen, Germany,
  Jun. 25--30, 2017, pp. 331--335.

\bibitem{8445979}
R.~Talak, S.~Karaman, and E.~Modiano, ``Distributed scheduling algorithms for
  optimizing information freshness in wireless networks,'' in \emph{Proc. IEEE
  Works. on Sign. Proc. Adv. in Wirel. Comms.}, Kalamata, Greece, Jun. 25--28,
  2018, pp. 1--5.

\bibitem{9007478}
A.~{Maatouk}, M.~{Assaad}, and A.~{Ephremides}, ``On the age of information in
  a {CSMA} environment,'' \emph{IEEE/ACM Trans. Net.}, Early Access 2020.

\bibitem{8877239}
M.~{Moltafet}, M.~{Leinonen}, and M.~{Codreanu}, ``Worst case analysis of age
  of information in a shared-access channel,'' in \emph{Proc. Int. Symp.
  Wireless Commun. Systems}, Oulu, Finland, Aug. 27--30, 2019, pp. 613--617.

\bibitem{rade2013mathematics}
L.~Rade and B.~Westergren, \emph{Mathematics Handbook for Science and
  Engineering}.\hskip 1em plus 0.5em minus 0.4em\relax Berlin, Germany:
  Springer, 2005.

\bibitem{ross2014introduction}
R.~Sheldon~M, \emph{Introduction to Probability Models}.\hskip 1em plus 0.5em
  minus 0.4em\relax California: Academic press, 2010.

\bibitem{Takacs1995}
L.~{Takacs}, ``Investigation of waiting time problems by reduction to {Markov}
  processes,'' \emph{Acta Math. Hung.}, vol.~6, pp. 101--129, 1955.

\bibitem{Kleinrock121}
L.~Kleinrock, \emph{Queueing Systems, Volume 1: Theory}.\hskip 1em plus 0.5em
  minus 0.4em\relax New York: John Wiley and Sons, 1995.

\bibitem{TCoMstabelqueue01}
P.~Cantrell, ``Computation of the transient {M/M/1} queue cdf, pdf, and mean
  with generalized {Q}-functions,'' \emph{{IEEE} Trans. Commun.}, vol.~34,
  no.~8, pp. 814--817, Aug. 1986.

\bibitem{bertsekas1992data}
D.~Bertsekas and R.~Gallager, \emph{Data Networks}.\hskip 1em plus 0.5em minus
  0.4em\relax Englewood Cliffs, New Jersey: Prentice-Hall International, 1992.

\bibitem{daigle2005basic}
J.~N. Daigle, \emph{Queueing Theory with Applications to Packet
  Telecommunication}.\hskip 1em plus 0.5em minus 0.4em\relax New York: Springer
  Science, 2005.

\bibitem{ourITW01}
M.~Moltafet, M.~Leinonen, and M.~Codreanu, ``Closed-form expression for the
  average age of information in a multi-source {M/G/1} queueing model,'' in
  \emph{Proc. IEEE Inform. Theory Workshop}, Visby, Gotland, Sweden, Aug.
  25--28, 2019.

\bibitem{papbook01}
A.~Papoulis, \emph{Probability, Random Variables, and Stochastic
  Processes}.\hskip 1em plus 0.5em minus 0.4em\relax New York: McGraw-Hill,
  1984.

\end{thebibliography}
\end{spacing}

\end{document}